\documentclass[11pt]{report}

\usepackage{amsmath}
\usepackage{amsfonts}
\usepackage{amssymb}
\usepackage{amsthm}
\usepackage[latin1]{inputenc}
\usepackage[spanish, activeacute]{babel}
\usepackage{listings}
\usepackage{hyperref}
\usepackage{soul}
\usepackage[pdftex]{color,graphicx}
\newtheorem{definicion}{Definición}[chapter]
\newtheorem{proposicion}{Proposición}[chapter]
\newtheorem{teorema}{Teorema}[chapter]
\newtheorem{corolario}{Corolario}[chapter]
\newtheorem{ejemplo}{Ejemplo}
\usepackage{xcolor}
\hypersetup{colorlinks,%
	    citecolor=black,%
	    filecolor=black,%
	    linkcolor=black,%
	    urlcolor=blue}
	    
\title{Diferenciación Automática Anidada. Un enfoque algebraico}
\author{Juan Luis Valerdi Cabrera}
\date{}

\begin{document}

\begin{titlepage}
	\begin{center}
		\includegraphics[width=2cm]{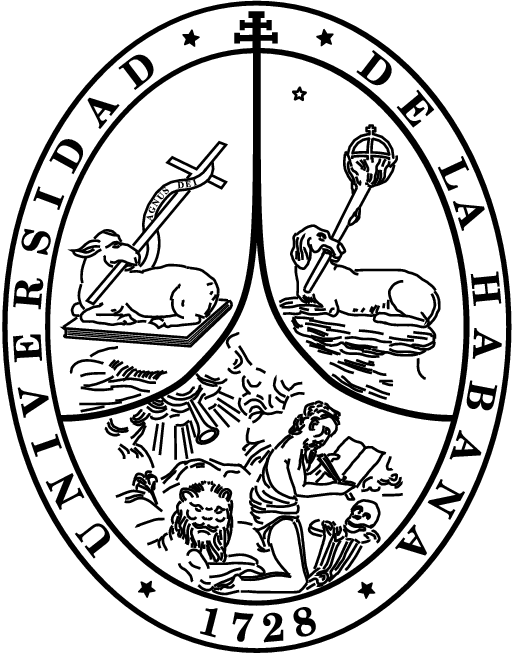}
	\end{center}
	\begin{center}
	Universidad de La Habana\\
	Facultad de Matemática y Computación
	\end{center}
	
	\vspace{0.7cm}
	
	\hrulefill
	
	\vspace{2pt}
	
	\begin{center}
		\huge \textbf{Diferenciación Automática Anidada.\\ Un enfoque algebraico.}
	\end{center}
	
	\vspace{2pt}
	
	\hrulefill

    \vspace*{0.75cm}
    
    \begin{center}
      \large
      Trabajo presentado como parte de los requisitos para optar por\\
      el t'itulo de Licenciado en Matemática
      
    \end{center}
    
	\vspace*{0.5cm}

	\begin{center}	
			{\Large \emph{Autor:}}\\
            {\Large Juan Luis Valerdi Cabrera}

            \vspace*{0.8cm}
            
			{\Large \emph{Tutor:}}\\
            {\Large \emph{Msc.} Fernando Raul Rodriguez Flores }\\[5pt]
	\end{center}
		
	\vspace{27pt}
	
	\begin{center}
		\large Junio de 2012
	\end{center}
\end{titlepage}
\selectlanguage{spanish} 

\begin{abstract}	

  En este trabajo se presenta una propuesta para realizar
  Diferenciación Automática Anidada utilizando cualquier biblioteca de
  Diferenciación Automática que permita sobrecarga de operadores.
  Para calcular las derivadas anidadas en una misma evaluación de la
  función, la cual se asume que sea anal'itica, se trabaja con el modo
  forward utilizando una nueva estructura llamada SuperAdouble, que
  garantiza que se aplique correctamente la diferenciación automática
  y se calculen el valor y la derivada que se requiera. También se
  presenta un enfoque algebraico de la Diferenciación Automática y en
  particular del espacio de los SuperAdoubles.
  \\
  
\end{abstract}

\tableofcontents

\chapter*{Introducción}\label{cha:intro}

 \addcontentsline{toc}{chapter}{Introducci'on}

 Muchas veces es necesario calcular el valor numérico de una función y
 a la vez obtener una aproximación precisa de las derivadas. Ejemplos de
 esto se pueden encontrar en los métodos numéricos de la programación
 no lineal \cite{luenberger}, en los métodos implícitos para la
 resolución numérica de ecuaciones diferenciales \cite{hairer} y en
 problemas inversos en la asimilación de datos \cite{griewank}.

 La Diferenciación Automática (AD, por sus siglas en inglés) es una
 herramienta que puede usarse para calcular las derivadas de cualquier
 función diferenciable, de forma automática. En este contexto
 ``automática''\ significa que el usuario solo necesita escribir el
 código fuente de la función en un lenguaje determinado, y la
 herramienta puede calcular las derivadas solicitadas sin incurrir en
 errores de truncamiento \cite{griewank}.

 Existen herramientas de diferenciación automática para varios lenguajes
 de programación, como ADOL-C \cite{adolc} para programas escritos 
 en C y C++, ADIFOR \cite{adifor} para programas escritos en Fortran, y 
 ADOLNET \cite{adolnet} para lenguajes de la plataforma .NET. 
 
 Una forma de crear estas herramientas de AD en lenguajes que soporten 
 sobrecarga de operadores es crear una clase llamada Adouble que contenga 
 dos campos de tipo double, uno llamado valor y otro derivada. Para esta 
 clase se sobrecargan los operadores aritm'eticos y las funciones elementales, 
 para que, al mismo tiempo que se calculen los resultados de las operaciones, 
 se calculen el valor de las derivadas de esas operaciones \cite{griewank}. 

 Sin embargo, existen ocasiones en las que es necesario calcular unas derivadas
 como paso intermedio para el cálculo de otras y con las t'ecnicas usuales de
 AD esto no es posible.  Un ejemplo sencillo donde se aprecia este fenómeno es 
 el siguiente.

 \begin{ejemplo} \label{ej:anidaad}
  Sea $f$ la función definida como  
  \begin{displaymath}
   f(x)= x^2+\dot{g}(x^3),
  \end{displaymath}
  donde $\dot{g}$ es la función derivada de $g(x)=e^{x^2}.$ Suponiendo que 
  sólo se conoce el código fuente de $g$ y de $f$, se desea calcular el valor y la 
  derivada de $f$ en el punto $x=x_0.$
 \end{ejemplo}

 Estas derivadas que intervienen en el cálculo de otras reciben
 el nombre de derivadas anidadas, y aunque existen herramientas 
 de Diferenciaci'on Autom'atica como ADOL-C \cite{riehme} que 
 permiten el cálculo de estas derivadas, no existe un método 
 general que funcione en todas las herramientas existentes \cite{siskind}, 
 por lo que puede decirse que no existe una solución general para 
 este problema. 
 
 En este trabajo se propone una metodología y una herramienta
 computacional para calcular derivadas anidadas de funciones
 anal'iticas utilizando cualquier biblioteca de AD que soporte
 sobrecarga de operadores.

 La metodolog'ia propuesta en este trabajo parte de la creación de un
 nuevo tipo de dato que se llamar'a SuperAdouble. Al igual que el tipo
 de dato Adouble, contendr'a dos campos, uno para almacenar el valor y
 otro para almacenar la derivada. A diferencia de los Adoubles, en el
 que los campos valor y derivada son de tipo double, en los
 Superadoubles, estos campos ser'an de tipo Adouble. Para esta nueva
 clase tambi'en se sobrecargan los operadores aritm'eticos y las
 funciones elementales para calcular simult'aneamente el valor de las
 operaciones y sus derivadas.  Al realizar operaciones con estos
 SuperAdoubles es posible calcular derivadas anidadas. El objetivo
 fundamental de este trabajo es demostrar esta afirmaci'on.
 
 Para demostrar que al realizar las operaciones con
 SuperAdoubles se obtienen los valores y las derivadas deseadas se
 realiza una presentación algebraica del espacio de los números
 Adoubles (que son el fundamento de la Diferenciación Automática) y de
 los números SuperAdoubles (que son el fundamento de este trabajo).

 Este documento está estructurado de la siguiente forma: en el
 Cap'itulo \ref{cha:preliminares} se presentan conceptos b'asicos de
 la diferenciaci'on autom'atica y del 'algebra que se usar'an en este
 trabajo; en el Cap'itulo \ref{cha:ad} se presentan los elementos
 básicos de la AD, sus modos e implementaciones; en el Capítulo
 \ref{cha:metodologia} se introduce una metodología para utilizar AD
 cuando haya que calcular derivadas anidadas; el Capítulo
 \ref{cha:adouble} se dedica a definir y explicar la estructura
 algebraica que sustenta la AD, y en el Capítulo
 \ref{cha:superadouble} se define $\mathbb{SA}$, una estructura
 algebraica que ayuda a demostrar que los números superadoubles pueden
 utilizarse para calcular derivadas anidadas.
 
\chapter{Preliminares} \label{cha:preliminares}

 En este cap'itulo se presentan las definiciones y propiedades necesarias
 para el desarrollo de este trabajo como derivadas anidadas, estructuras
 algebraicas y elementos de lenguajes de programaci'on. En la Secci'on \ref{sec:nombres} 
 se presentan las definiciones relacionadas con las derivadas anidadas y en la 
 Secci'on \ref{sec:algebra}, las estructuras algebraicas que se usar'an en los 
 Cap'itulos \ref{cha:adouble} y \ref{cha:superadouble}.

 \section{Derivadas anidadas} \label{sec:nombres}
 
 En esta secci'on se definen algunos conceptos relacionado con la diferenciaci'on 
 anidada como variable original, variable anidada, funci'on anidada, derivada
 original, derivada anidada y derivada compuesta.
 
 \begin{definicion}
  Sea $f(x)=\dot{g}(h(x))$ una funci'on real donde $g$ y $h$ son funciones
  conocidas. Se les
  llamar'a a las diferentes partes de $f$ de la siguiente forma:
  \begin{itemize}
	 \item A $x$ se le llamar'a variable original.
	 \item A $y = h(x)$ se le llamar'a variable anidada.
	 \item A $g(x)$ se le llamar'a funci'on anidada.
	 \item A $g'(x)=\frac{dg(h(x))}{dx}$ se le llamar'a derivada original.
	 \item A $\dot{g}(x)=\frac{dg(h(x))}{dh}$ se le llamar'a derivada anidada.
	 \item A $\dot{g}'(x)=\frac{d\dot{g}(h(x))}{dx}=\frac{d^2g(h(x))}{dxdh}$ se le llamar'a derivada compuesta.  
	\end{itemize}
 \end{definicion}
 
 En el caso de que $h(x)=x$ los conceptos de derivada original y derivada
 anidada coinciden, lo cual implica que la derivada compuesta sea 
 $\frac{d^2g(x)}{dx^2}$.
 
 \begin{ejemplo}
  La funci'on del Ejemplo \ref{ej:anidaad} es
  \begin{displaymath}
   f(x)= x^2+\dot{g}(x^3),	
  \end{displaymath} 
  con $g(x)=e^{x^2}.$
  En este caso se tiene lo siguiente:
  \begin{itemize}
 	 \item La variable $x$ es la variable original. 
	 \item $y=x^3$ es la variable anidada.
	 \item La funci'on $g(x)=e^{x^2}$ es la funci'on anidada.
	 \item La derivada original es $g'(x)=6x^5e^{x^6}$.  
	 \item La derivada anidada es $\dot{g}(x)=2x^3e^{x^6}$.
	 \item La derivada compuesta es $\dot{g}'(x)=6x^2e^{x^6}(1+2x^6)$.
  \end{itemize}
 \end{ejemplo}
 
 \section{Estructuras algebraicas} \label{sec:algebra}
 
 En esta secci'on se presentan las definiciones y resultados del 'algebra abstracta que se utilizar'an en los Cap'itulos \ref{cha:adouble}
 y \ref{cha:superadouble}. Este contenido se puede encontrar en \cite{queysanne} 
 y \cite{hazewinkel}.
 
 \begin{definicion}
  Un conjunto $G$ provisto de una ley interna $+$, es un grupo si $\forall a,b\in G$ 
  se cumple las siguiente propiedades:
  \begin{enumerate}
	 \item La suma es asociativa: $(a+b)+c=a+(b+c)$ 
	 \item Existe un elemento neutro: $\exists e\in G: a+e=e+a=a$
	 \item Cada elemento posee opuesto: $\exists a'\in G:a+a'=a'+a=e$
  \end{enumerate}
 \end{definicion}
 
 Al elemento $e$ de la definici'on anterior se le llama neutro o cero del grupo.
 En el caso de que la operaci'on $+$ sea conmutativa, se dice que el grupo es abeliano.
 
 \begin{definicion} \label{def:anillo}
  Un conjunto $A$ provisto de una adici'on $+$ y de una multiplicaci'on $\cdot$ (el s'imbolo del operador se puede omitir) que sean
  internas posee una estructura de anillo si $\forall a,b,c\in A$ se cumple:
  \begin{enumerate}
	 \item $A$ posee una estructura de grupo abeliano para la adici'on
	 \item El producto es asociativo: $(ab)c=a(bc)$
	 \item El producto distribuye respecto a la suma a la izquierda: $a(b+c)=ab+ac$ 
	 \item El producto distribuye respecto a la suma a la derecha: $(b+c)a=ba+ca$
  \end{enumerate}
 \end{definicion}
 
 Si la multiplicaci'on posee un elemento neutro $e$ se dice que el anillo es
 unitario. Siendo $A$ un anillo unitario, si dado $a\in A$ existe un $a'$
 tal que $aa'=a'a=e$, se dice que $a$ es inversible en $A$.
 
 Si la multiplicaci'on en un anillo $A$ es conmutativa se dice que $A$ es
 un anillo conmutativo.
 
 \begin{definicion}
  Se llama subanillo de un anillo $A$ a un subconjunto no vac'io $B$ de $A$ 
  con las operaciones de $A$ internas en $B$ que hacen de $B$ un anillo.
 \end{definicion}
 
 \begin{teorema}
  Para que un subcojunto no vac'io $B$ de un anillo $A$ sea un subanillo de $A$
  es necesario y suficiente que para todo $a\in A$ y $b\in B$, entonces 
  \begin{displaymath}
    a-b\in B \wedge ab\in B.
  \end{displaymath}
 \end{teorema}
 \begin{proof}
  V'ease en \cite{queysanne}.
 \end{proof}
 
 \begin{teorema} \label{teo:binomio}
  En un anillo conmutativo se cumple 
  \begin{equation}
   (a_1+a_2+ \ldots +a_m)^n=
   \sum{\frac{n!}{p_1!p_2!\ldots p_m!}a_1^{p_1}\ldots(a_m)^{p_m}}, 
  \end{equation}
  donde la suma del segundo miembro se extiende a todas las combinaciones de 
  $p_1, p_2,\ldots ,p_m$ de enteros no negativos tales que
  \begin{displaymath}
   \sum_{i=1}^mp_i=n.	
  \end{displaymath}
 \end{teorema}
 \begin{proof}
  V'ease en \cite{queysanne}.
 \end{proof}
 
 \begin{definicion}
  Un conjunto $K$ provisto de una adici'on y de una multiplicaci'on posee
  una estructura de cuerpo para esas dos operaciones si:
  \begin{enumerate}
	 \item $K$ posee una estructura de anillo para esas dos operaciones.
	 \item $L^*=K-\{0\}$ (donde $0$ es el elemento neutro de la adici'on) posee una
	       estructura de grupo para la multiplicaci'on.
  \end{enumerate}
 \end{definicion}
  
 Si la multiplicaci'on en un cuerpo $K$ es conmutativa se dice que $K$ es
 un cuerpo conmutativo.
 
 \begin{definicion}
  Dado un cuerpo conmutativo $K$, de elementos neutros $0$ y $1$ con
  respecto a la suma y la multiplicaci'on, se dice que un conjunto $E$
  provisto de una operaci'on interna y de una operaci'on externa cuyo
  dominio de operadores es $K$, tiene una estructura de espacio vectorial
  sobre $K$ si $\forall\lambda,\mu\in K$ y $\forall x,y\in E$ se cumple:
  \begin{enumerate}
	 \item $E$ es un grupo abeliano para sus operaci'on interna
	 \item La operaci'on externa cumple: $\lambda(\mu x)=(\lambda\mu)x$
	 \item La operaci'on externa cumple: $1x=x$
	 \item La operación externa es distributiva con relación a la suma en $K: (\lambda+\mu)x=\lambda x+\mu x$
	 \item La operación externa es distributiva con relaci'on a la operaci'on interna de $E: \lambda(x+y)=\lambda x+ \lambda y$
  \end{enumerate}
 \end{definicion}
  
 \begin{definicion}
  Un conjunto $A$ provisto de una estructura vectorial sobre un cuerpo $K$ 
  y de una estructura de anillo se dice que posee una estructura de
  'algebra sobre $K$ si $\forall \alpha,\beta\in K$ y $\forall x,y\in A$
  se cumple
  \begin{displaymath}
	 (\alpha x)(\beta y) = (\alpha\beta)(xy).
  \end{displaymath}  
 \end{definicion} 
  
 Estos conceptos y resultados se usar'an en los Cap'itulos \ref{cha:adouble} y
 \ref{cha:superadouble} para la representaci'on algebraica de la AD.

\chapter{Diferenciación Automática}\label{cha:ad}

 En este cap'itulo se realiza una introducci'on a la diferenciaci'on
 autom'atica, en la que se presentar'an sus modos de aplicaci'on, ejemplos y
 v'ias de implementaci'on. Para la implementaci'on se presentar'an c'odigos
 en el lenguaje de programaci'on C\# con el objetivo de mostrar la aplicaci'on de la AD 
 en los ejemplos que se presentan en este cap'itulo.

 \section{Introducción a la AD}
  La diferenciación automática, también conocida como 
  diferenciación algorítmica, es un conjunto de técnicas y herramientas 
  que permiten evaluar numéricamente la derivada de una función definida 
  mediante su código fuente en un lenguaje de programación.  

  La AD está basada en el hecho de que para evaluar una función en un
  lenguaje de programación dado se ejecutan una secuencia de operaciones
  aritméticas (adición, sustracción, multiplicación y división) y llamados a
  funciones elementales (exp, log, sen, cos, etc.). Aplicando la regla de
  la cadena al mismo tiempo que se realizan estas operaciones, se pueden 
  calcular derivadas de cualquier orden tan exactas como la aritmética 
  de la máquina lo permita \cite{griewank}.

  La base de la AD es la descomposición de diferenciales que provee la regla 
  de la cadena. Para una composición de funciones $f(x) = g(h(x))$ se obtiene, 
  a partir de la regla de la cadena,
  
  \begin{displaymath}
	 \frac{df}{dx} = \frac{dg}{dh} \frac{dh}{dx}.
  \end{displaymath}

  Existen dos modos para aplicar la AD, el modo hacia adelante 
  o modo forward y el modo hacia atrás o modo reverse. El modo forward se obtiene 
  al aplicar la regla de la cadena de derecha a izquierda, 
  es decir, primero se calcula $dh / dx$ y después $dg / dh$, mientras que el modo 
  reverse se obtiene cuando se aplica la regla de la cadena de izquierda a derecha \cite{griewank}.  
 
  A continuaci'on se presentan las ideas fundamentales del modo hacia adelante,
  que por ser m'as sencillo e intuitivo, permite explicar con mayor claridad
  el funcionamiento de la diferenciaci'on autom'atica. El lector interesado
  en el modo hacia atr'as puede consultar \cite{griewank}.

  \subsection{Modo forward}

  Los siguientes ejemplos ilustran el funcionamiento del modo forward:

  \begin{ejemplo} \label{ej:ad1}
   Se desea calcular el valor y la derivada de 
   \begin{displaymath}
    f(x) = x^2\cdot\cos(x) 	
   \end{displaymath}
   en el punto $x=\pi$.
  \end{ejemplo}
  
  La siguiente tabla contiene las operaciones necesarias para evaluar esta funci'on
  en una computadora.
  
  \begin{displaymath}
	 \begin{tabular}{|ccccc|}
    \hline
    $w_1$ & = & $x$ & = & $\pi$ \\
    \hline
    $w_2$ & = & $w_1^2$ & = & $\pi^2$ \\
    $w_3$ & = & $\cos(w_1)$ & = & $-1$ \\
    $w_4$ & = & $w_2w_3$ & = & $-\pi^2$ \\
    \hline
    $y$ & = & $w_4$ & = & $-\pi^2$ \\
    \hline
   \end{tabular}
  \end{displaymath}

  Para aplicar el modo hacia adelante, se almacena en una nueva variable
  la derivada de cada operaci'on con respecto a la variable
x. Esta derivada se puede calcular al mismo tiempo 
  que se realiza la operaci'on. La siguiente tabla ilustra este procedimiento.
  
  \begin{displaymath}
	 \begin{tabular}{|ccccc|ccccc|}
    \hline
    $w_1$ & = & $x$ & = & $\pi$ & $\dot{w}_1$ & = & $\dot{x}$ & = & $1$ \\
    \hline
    $w_2$ & = & $w_1^2$ & = & $\pi^2$ & $\dot{w}_2$ & = & $2w_1\dot{w}_1$ & = & $2\pi$ \\
    $w_3$ & = & $\cos(w_1)$ & = & $-1$ & $\dot{w}_3$ & = & $\sin(w_1)\dot{w}_1$ & = & $0$\\
    $w_4$ & = & $w_2w_3$ & = & $-\pi^2$ & $\dot{w}_4$ & = & $w_2\dot{w}_3+\dot{w}_2w_3$ & = & $-2\pi$\\
    \hline
    $y$ & = & $w_4$ & = & $-\pi^2$ & $\dot{y}$ & = & $\dot{w}_4$ & = & $-2\pi$ \\
    \hline
   \end{tabular}
  \end{displaymath}

  En la tabla anterior, las variables $\dot{w}_i$ almacenan la 
  derivada con respecto a x de la operaci'on $w_i$. Como 
  se desea calcular la derivada con respecto a la variable $x$, entonces la nueva variable
  $\dot{w}_1=\frac{dw_1}{dx}$ se inicializa con el valor $1$. Al finalizar el c'alculo, en la variable 
  $\dot{w}_4$ se tiene el valor de la derivada de $f(x)$ con respecto a $x$.
  
  El siguiente ejemplo muestra c'omo se puede aplicar el modo hacia adelante para calcular derivadas
  parciales de funciones de m'as de una variable.

  \begin{ejemplo} \label{ej:ad}
   Se desea calcular el valor y el gradiente de 
   \begin{displaymath}
	  f(x_1,x_2) = x_1x_2+\sen(x_1)
   \end{displaymath}
   en $x_1=\pi$ y $x_2=2.$
  \end{ejemplo}
  
  La función $f$ de este ejemplo puede expresarse mediante las siguientes operaciones.
  \begin{displaymath}
	 \begin{tabular}{|ccccc|}
    \hline
    $w_1$ & = & $x_1$ & = & $\pi$ \\
    $w_2$ & = & $x_2$ & = & $2$ \\
    \hline
    $w_3$ & = & $w_1w_2$ & = & $2\pi$ \\
    $w_4$ & = & $\sen(w_1)$ & = & $0$ \\
    $w_5$ & = & $w_3+w_4$ & = & $2\pi$ \\
    \hline
    $y$ & = & $w_5$ & = & $2\pi$ \\
    \hline
   \end{tabular}
  \end{displaymath}
  
  Para calcular el gradiente utilizando el modo forward hay que realizar el mismo 
  procedimiento del ejemplo anterior, pero en este caso dos veces: una por cada
  derivada parcial. A continuaci'on se muestra el c'alculo de la derivada parcial
  $\frac{df}{dx_1}$. 
  
  \begin{displaymath}
	 \begin{tabular}{|ccccc|ccccc|}
    \hline
    $w_1$ & = & $x_1$ & = & $\pi$ & $\dot{w}_1$ & = & $\dot{x}_1$ & = & $1$ \\
    $w_2$ & = & $x_2$ & = & $2$ & $\dot{w}_2$ & = & $\dot{x}_2$ & = & $0$ \\
    \hline
    $w_3$ & = & $w_1w_2$ & = & $2\pi$ & $\dot{w}_3$ & = & $\dot{w}_1w_2+w_1\dot{w}_2$ & = & $2$ \\
    $w_4$ & = & $\sen(w_1)$ & = & $0$ & $\dot{w}_4$ & = & $\cos(w_1)\dot{w}_1$ & = & $-1$ \\
    $w_5$ & = & $w_3+w_4$ & = & $2\pi$ & $\dot{w}_5$ & = & $\dot{w}_3+\dot{w}_4$ & = & $1$ \\
    \hline
    $y$ & = & $w_5$ & = & $2\pi$ & $\dot{y}$ & = & $\dot{w}_5$ & = & $1$ \\
    \hline
   \end{tabular}
  \end{displaymath}

  En este caso, la variable $\dot{w}_2$ se inicializa con valor $0$ porque
  $\frac{\partial x_2}{\partial x_1}=0, w_2=x_2$ y $\dot{w}_2=\frac{\partial w_2}{\partial x_1}$.
  Al finalizar los c'alculos, en la variable $\dot{w}_5$ se tiene el valor
  de la derivada $\frac{\partial f(x)}{\partial x_1}$.

  El procedimiento presentado en los ejemplos anteriores se puede generalizar 
  de la siguiente forma. 
  
  Sea $f:R^n\rightarrow R^m$ una función diferenciable. Se denotarán las $n$ 
  variables reales de entrada como
  \begin{displaymath}
	 w_{i-n}=x_i\qquad \textrm{y}\qquad \dot{w}_{i-n}=\dot{x}_i,\ \ 
	 \textrm{con } i=1\ldots n \textrm{\ \ y \ } \dot{x}_i=1.
  \end{displaymath}
  Las variables $\dot{x}_i$ se inicializan con el valor de las
  derivadas parciales de cada variable $x_i$ con respecto a la variable original.
  Si se desea calcular la derivada parcial de $f$ respecto a $x_i$, entonces 
  $\dot{x}_i$ deber'ia ser $1$ y $\dot{x}_j$ deber'ia ser $0$ para toda $j$ 
  diferente de $i$.
  
  Como $f$ está compuesta por funciones elementales, es conveniente denotarlas de alguna forma. 
  A la $j$-ésima función elemental se le denotará por $\phi_j$. 
  
  Para descomponer a $f$ en operaciones 
  elementales se denotarán nuevas variables $w_i$ y $\dot{w}_i$ como
  
  \begin{equation}
   w_i=\phi_i(w_j)\qquad\textrm{con\ }j\prec i, \label{eq:notaciondependencia}
  \end{equation}
  
  \begin{displaymath}
	 \dot{w}_i=\sum_{j\prec i}\frac{\partial}{\partial w_j}\phi_i(w_j)\dot{w}_j,
  \end{displaymath}

  donde en este caso $i=1\ldots l$ y $l$ es el número de operaciones elementales que componen a $f$. 
  La simbología $\phi_i(w_j)$ con $j\prec i$ significa que $\phi_i$ depende directamente de $w_j$, 
  es decir, que para evaluar $\phi_i$ se usa explícitamente $w_j$. También es usual denotar 
  (\ref{eq:notaciondependencia}) como
  
  \begin{displaymath}
	 w_i=\phi_i(w_j)_{j\prec i}.
  \end{displaymath}

  Con las notaciones anteriores se puede expresar el procedimiento
  general de la siguiente forma:

	\begin{displaymath}
	 \begin{tabular}{|l l|}
	  \hline
	  $w_i=x_i$ & \\
	  & $i=1\ldots n$\\
	  $\dot{w}_{i-n}=\dot{x}_i$ & \\
	  \hline
	  $w_i=\phi_i(w_j)_{j\prec i}$ & \\
	  & $i=1\ldots l$\\
	  $\dot{w}_i=\sum_{j\prec i}\frac{\partial}{\partial w_j}\phi_i(w_j)\dot{w}_j$ & \\
	  \hline
	  $y_{m-i}=w_{l-i}$ & \\
	  & $i=m-1\ldots 0$\\
	  $\dot{y}_{m-i}=\dot{w}_{l-i}$ & \\
	  \hline
	 \end{tabular}
  \end{displaymath}
  
  Como se puede apreciar en la tabla anterior, el procedimiento general del modo forward 
  est'a estructurado en tres fases: primero se inicializan las variables, despu'es se 
  ejecutan las operaciones y se calculan las derivadas, y finalmente se devuelven los 
  valores y la derivada calculada.
  
  En esta secci'on se ha presentado la AD desde un punto de vista te'orico. En la siguiente
  secci'on se muestra una posible v'ia para implementar estas ideas computacionalmente.

 \section{Implementación} \label{sec:implementacion}
 
 Existen dos estrategias para lograr AD: transformación del 
 código fuente y sobrecarga de operadores \cite{griewank}.

 Para utilizar la metodolog'ia propuesta en este trabajo resulta m'as conveniente 
 utilizar la sobrecarga de operadores, por lo que en esta secci'on se presentan sus
 elementos fundamentales. El lector interesado en la transformaci'on de c'odigo puede
 consultar \cite{griewank1}. 

 Para implementar la AD utilizando sobrecarga de operadores se debe crear una nueva clase, 
 la cual se llamar'a en este trabajo Adouble siguiendo la notaci'on de \cite{griewank}.
  
 En esta clase se deben definir dos campos de números reales, uno, que usualmente recibe 
 el nombre de valor, para almacenar el resultado de la operaci'on que este Adouble
 representa, y otro, que usualmente recibe el nombre de derivada, para almacenar la
 derivada de esa operaci'on. Estos campos valor y derivada son la representaci'on 
 computacional de las variables $w_i$ y $\dot{w}_i$. 
  
 Una vez definida esta nueva clase, se sobrecargan
 los operadores aritm'eticos y las funciones elementales, para que, al mismo tiempo que
 se calculan los resultados de las operaciones, tambi'en se pueda calcular el valor de 
 las derivadas.
 
 La siguiente tabla muestra los valores de los campos valor y derivada
 de cada uno de los adoubles que intervienen en la evaluaci'on de la funci'on
 $f(x) = x^2\cdot{}cos(x)$. N'otese que cuando se realizan todas las operaciones, en el
 campo derivada del adouble $y$ se obtiene el valor de la derivada de la funci'on
 en el punto en que fue evaluada.
 
 \begin{displaymath}
  \begin{tabular}{|l|l|l|}
   \hline
   Adouble $w_1$ & $w_1$.valor $= \pi$ & $w_1$.derivada $= 1$ \\
   \hline
   Adouble $w_2 = w_1^2$ & $w_2$.valor $=\pi^2$ & $w_2$.derivada $= 2\pi$ \\
   Adouble $w_3 = \cos(w_1)$ & $w_3$.valor $=-1$ & $w_3$.derivada $=0$\\
   Adouble $w_4 = w_2*w_3$ & $w_4$.valor $=-\pi^2$ & $w_4$.derivada $=-2\pi$\\
   \hline
   Adouble\ \ \ $y = w_4$ &\ \ $y$.valor\ $=-\pi^2$ &\ \  $y$.derivada\ $=-2\pi$\\
   \hline
  \end{tabular}
 \end{displaymath} 
  
 Esta vía tiene la ventaja de que es f'acil de implementar y que para utilizarlo solo hay
 que modificar ligeramente el c'odigo fuente de la funci'on que se desea derivar \cite{griewank1}.

 En la Figura \ref{cod1} se muestra la implementaci'on de un programa en C\# que calcula el valor de 
 $f(x) = x^2\cdot\cos(x)$ en $x = 2$; y en la Figura \ref{cod2} las modificaciones necesarias para calcular la derivada en ese punto.
 
 \begin{figure}[htb]
  \begin{verbatim}
class Example
{
  static void Main()	
  {
    double x = 2; //Inicialización de x
                  //para evaluar f en 2
    double y = x*x + cos(x); //Cálculo de f
    Console.WriteLine("El valor de f en x = 2 es: " + y);
  }
}
  \end{verbatim}
	\caption{Código de $f(x) = x^2\cdot\cos(x)$}\label{cod1}
 \end{figure}
 
 \begin{figure}[htb]
  \begin{verbatim}
class Example
{
  static void Main()	
  {
    Adouble x = new Adouble(); \\Inicialización de x como 
                               \\Adouble
    x.valor = 2; \\ Se quiere el valor y la derivada
                 \\ en x = 2
    x.derivada = 1; \\La derivada de x con respecto a x es 1
    Adouble y = x*x + cos(x);
    Console.WriteLine("El valor de f en x = 2 es: " + w4.valor);
    Console.WriteLine("La derivada de f en x = 2 es: " + w4.derivada);
  }
}
  \end{verbatim}
	\caption{Código modificado de $f(x) = x^2\cdot\cos(x)$}\label{cod2}
 \end{figure}
 
 Esta idea se puede modificar f'acilmente para calcular derivadas parciales 
 de funciones de m'as de una variable. El lector interesado en este tema puede consultar
 \cite{griewank} y \cite{adolc}.
 
 En este cap'itulo se han presentado las ideas fundamentales de la AD 
 en las que se sustenta la soluci'on propuesta en el cap'itulo siguiente 
 para el c'alculo de derivadas anidadas utilizando diferenciaci'on autom'atica.
  
\chapter{C'alculo de Derivadas Anidadas} \label{cha:metodologia}

 En este cap'itulo se presenta una metodolog'ia para calcular derivadas anidadas
 utilizando cualquier biblioteca de Diferenciaci'on Autom'atica mediante sobrecarga 
 de operadores. Esta metodolog'ia se presenta para funciones de una variable real
 con el objetivo de facilitar su exposici'on, pero su extensi'on a funciones de varias variables
 es posible de la misma forma que las t'ecnicas de diferenciaci'on autom'atica se 
 extienden a funciones varias variables.

 La metodolog'ia que se propone es la siguiente: 
 
 \begin{enumerate}
  \item Partir de una biblioteca de Diferenciaci'on Autom'atica en la que exista
        un tipo de dato Adouble.
  \item Definir un nuevo tipo de dato llamado SuperAdouble. Esta nueva estructura tendr'a 
        dos campos: valor y derivada, y ambos campos ser'an del tipo de dato Adouble
        definido en la biblioteca del paso 1. El hecho de que 
        estos campos sean de tipo Adouble permitir'a calcular las derivadas anidadas.
  \item Definir funciones para construir superadoubles a partir de adoubles y para obtener
        los campos valor y derivada de variable de tipo SuperAdouble. Estas funciones se presentan en la 
        Secci'on \ref{paso2}.
  \item Modificar el c'odigo fuente de la funci'on anidada para calcular las derivadas anidadas.
 \end{enumerate}      
 
 En las siguientes secciones se detallan cada uno estos pasos.
 
 \section{La clase SuperAdouble} \label{paso1}
 
 El segundo paso propone crear una nueva clase llamada SuperAdouble. Este nuevo tipo de dato tiene un
 campo valor y un campo derivada, ambos de tipo Adouble. Al igual que la clase Adouble, los operadores
 aritm'eticos y las funciones elementales se sobrecargan para que cuando operen con valores de este tipo
 permitan calcular los valores y las derivadas anidadas que se deseen. 
          
 El campo valor de los SuperAdouble es de tipo Adouble, por lo que este contiene
 a un campo valor y un campo derivada, ambos de tipo double. Es decir, 
 para una variable de tipo SuperAdouble tiene sentido hablar de los campos 
 valor.valor y valor.derivada. Por el mismo
 motivo, en una variable de tipo superadouble existen los campos derivada.valor y derivada.derivada.
 
 Estos cuatro campos: valor.valor, valor.derivada, derivada.valor y derivada.
 derivada se relacionan con las derivadas anidadas a trav'es del siguiente
 resultado, que es el principal aporte de este trabajo.
 
 \begin{teorema} \label{teo:teorema}
  En un objeto de tipo SuperAdouble se cumple que:
  \begin{itemize}
 	 \item valor.valor es el valor de la funci'on anidada. 
	 \item valor.derivada es la derivada original.
	 \item derivada.valor contiene la derivada anidada.
	 \item derivada.derivada contiene la derivada compuesta.
  \end{itemize}
 \end{teorema}
 \begin{proof}
  El objetivo del Cap'itulo \ref{cha:superadouble} es demostrar este teorema.
 \end{proof}
 
 Una vez que est'e definida esta nueva clase, es necesario definir funciones
 para crear instancias de esta clase y obtener las derivadas y valores deseados.
 Estas funciones se muestran en la siguiente secci'on.
              
 \section{Relaci'on entre Adoubles y Superadoubles.} \label{paso2}
              
 Para calcular derivadas anidadas son necesarias tres funciones que relacionen
 variables de tipo Adouble y variables de tipo SuperAdouble. Estas funciones 
 son \emph{push}, \emph{popV} y \emph{popD}. A continuaci'on se presentan 
 cada una de ellas y su relaci'on con el c'alculo de derivadas anidadas.
 
 La función \emph{push} recibe como par'ametro una variable $x$ de tipo Adouble y devuelve 
 una variable $X$ de tipo SuperAdouble. El campo valor de la
 variable $X$ ser'ia el Adouble $x$, y el campo derivada ser'ia un Adouble con valor
 1 y derivada 0. 
 
 \begin{ejemplo}
  Si se tiene un Adouble $x = (4, 10)$, donde la primera componente del vector es el campo valor del
  Adouble, y la segunda, el campo derivada, entonces al aplicar push
  se obtiene un SuperAdouble $X$ con los siguientes campos:
  \begin{itemize}
 	 \item $X.$valor.valor = 4.
	 \item $X.$valor.derivada = 10.
	 \item $X.$derivada.valor = 1.
	 \item $X.$derivada.derivada = 0.
  \end{itemize}
 \end{ejemplo}
 
 Por otra parte, las funciones \emph{popV} y \emph{popD} reciben como argumento un 
 SuperAdouble y devuelven un Adouble. La funci'on \emph{popV} devuelve el campo valor
 del SuperAdouble y \emph{popD} devuelve su campo derivada. 
 
 \begin{ejemplo}
  Si se aplica la funci'on popV al Superadouble $X$ definido en el ejemplo
  anterior se obtiene un Adouble $x$ con los siguientes campos:
  \begin{itemize}
   \item $x.$valor = $4$.
	 \item $x.$derivada = $10$.
  \end{itemize} 
 
  Por otra parte, si se aplica la funci'on popD al
  mismo SuperAdouble $X$ se obtiene:
  \begin{itemize}
	 \item $x.$valor = $1$.
	 \item $x.$derivada = $0$.
  \end{itemize} 
 \end{ejemplo}
 
 Para calcular derivadas anidadas, estas funciones deben usarse para modificar el 
 c'odigo fuente de la funci'on anidada como se muestra en la siguiente
 secci'on.
  
 \section{Modificaci'on del c'odigo de la funci'on anidada}
 
 En esta secci'on se presentan las modificaciones que son necesarias realizar
 en el c'odigo fuente de la funci'on anidada para utilizar los SuperAdouble, las
 funciones \emph{push, popD} y \emph{popV} y obtener los valores de las derivadas anidadas.
  
 Siguiendo con el Ejemplo \ref{ej:anidaad},
 \begin{displaymath}
   f(x)= x^2+\dot{g}(x^3), \qquad g(x)=e^{x^2}, 
 \end{displaymath}	
 la Figura \ref{fig:modif} muestra el c'odigo que permite calcular 
 el valor de la funci'on y su derivada en el punto $x=3$ usando los SuperAdoubles y las funciones 
 presentadas en la secci'on anterior.
 
 \begin{figure}[htb]
  \begin{verbatim}
class Example
{
  static void Main()	
  {
    Adouble x = new Adouble();
    x.valor = 3;
    x.derivada = 1;
    Adouble w1 = x*x;             \\ w1 = x^2
    Adouble w2 = w1*x;            \\ w2 = x^3
    SuperAdouble Y1 = push(w2);   \\ crear el superadouble
    SuperAdouble Y2 = exp(Y1*Y1); \\ g(x) = exp(x^2)
    SuperAdouble W3 = exp(W2);
    Adouble w3 = popD(Y2);        \\ obtener la derivada anidada
    Adouble w4 = w1 + w3;         \\ w4 = x^2 + gdot(x^3)
    Console.WriteLine("El valor de f en x = 3 es: " + 
                       w5.valor);
    Console.WriteLine("La derivada de f en x = 3 es: " + 
                       w5.derivada);
  }
}  
  \end{verbatim}
  \caption{Aplicaci'on de los SuperAdoubles en la AD.} \label{fig:modif}
 \end{figure}
 
 En el c'odigo de la Figura \ref{fig:modif} se empieza trabajando con
 números de tipo Adouble para calcular el valor y la derivada de las
 operaciones que no pertenezcan a la funci'on anidada. La l'inea
 \begin{verbatim}
  SuperAdouble Y1 = push(w2);
 \end{verbatim}
 \vspace*{-0.5cm}
 crea un objeto SuperAdouble a partir de la variable anidada. Despu'es
 de la creaci'on de W1 se efect'uan las operaciones de la funci'on anidada utilizando
 Superadoubles. Con la l'inea
 \begin{verbatim}
  Adouble w3 = popD(Y2);
 \end{verbatim}
 \vspace*{-0.5cm}
 se obtiene un Adouble que en su campo valor tiene la
 derivada anidada, y en su campo derivada tiene la derivada
 compuesta. Finalmente, en la l'inea
 \begin{verbatim}
  Adouble w4 = w1 + w3;
 \end{verbatim}
 \vspace*{-0.5cm} se termina el c'alculo de la funcion original.
 
 Como todas las operaciones se han realizado con adoubles, en el campo
 derivada del Adouble w4 se tiene la derivada de la funci'on original, que era
 el objetivo que se persegu'ia.
  
 Este ejemplo ilustra que usando la metodolog'ia propuesta es posible
 calcular derivadas anidadas. En los pr'oximos cap'itulos se
 fundamentan los resultados de esta secci'on desde una representaci'on
 algebraica y se demuestra Teorema \ref{teo:teorema} que es el
 basamento matem'atico de la metodolog'ia.
 
\chapter{Espacio Adouble}\label{cha:adouble}

 En este cap'itulo se presenta la diferenciaci'on autom'atica desde un punto 
 de vista algebraico. Para representar la clase Adouble mediante una estructura 
 algebraica se extienden los n'umeros reales agreg'andole una nueva componente 
 y definiendo una aritm'etica que garantice que en la primera componente se 
 obtengan los resultados de cada operaci'on, y en la segunda, las derivadas 
 de estas operaciones. 
 
 Esta presentaci'on de la AD desde el 'algebra no es nueva. La estructura
 algebraica que representa a los n'umeros adoubles se conoce en la literatura
 como n'umeros duales y tiene aplicaciones en la f'isica y la mec'anica \cite{cheng}. 
 
 En este cap'itulo se presentan las principales propiedades de estos n'umeros
 duales, que en este trabajo se les llamar'a
 n'umeros adoubles, o simplemente adoubles. Al conjunto de todos los adoubles
 se les llamar'a espacio Adouble.
 
 El objetivo fundamental de presentar estas propiedades y sus demostraciones,
 es que sirven de base para, en el Cap'itulo \ref{cha:superadouble}, presentar 
 el espacio SuperAdouble, sus propiedades fundamentales y demostrar que 
 utiliz'andolas se pueden calcular derivadas anidadas. 
   
 La estructura de este cap'itulo es la siguiente: en la Secci'on \ref{sec:ad1}
 se definir'a el espacio Adouble y se demostrar'a algunas de sus propiedades
 elementales; y en la Secci'on \ref{sec:ad3}
 se definen funciones de variable adouble y se demuestra que al evaluar una funci'on anal'itica en un
 adouble, se obtiene un nuevo adouble que en su primera componente tiene el
 valor de la evaluaci'on de la funci'on, y en la segunda, el valor de la derivada
 de la funci'on.
 
 \section{Definición de $\mathbb{A}$. Operaciones elementales} \label{sec:ad1}

 Observando c'omo funciona la AD desde el punto de vista computacional, es posible definir
 una estructura algebraica que represente la clase Adouble. Esta se crea expandiendo el 
 álgebra de los números reales con otra componente y definiendo una nueva aritmética
 que permita obtener, en
la segunda componente las derivadas de las operaciones, como se
muestra en la siguiente definici'on.
 
 \begin{definicion}
  El espacio Adouble $\mathbb{A}$ es el conjunto de los pares ordenados $(x,\dot{x})$ de números reales  
  con la suma y multiplicación definidos por
  \begin{displaymath}
	 (a,\dot{a})+(b,\dot{b})=(a+b,\dot{a}+\dot{b}),
  \end{displaymath}
  \begin{displaymath}
	 (a,\dot{a})(b,\dot{b})=(ab,a\dot{b}+\dot{a}b).
  \end{displaymath}
 \end{definicion}
 
 A los elementos del espacio Adouble se les llamar'a números adoubles o simplemente adoubles.
 
 En las operaciones de suma y multiplicaci'on definidas en este espacio se aprecia que en la 
 primera componente se tiene el c'alculo usual sobre los reales y en la segunda componente se 
 obtiene la aritm'etica de las derivadas. 
 
 A partir de estas dos operaciones es posible obtener la resta y la divisi'on as'i como funciones 
 que al ser evaluadas en un adouble devuelvan otro adouble que en la primera componente tengan el 
 resultado de evaluar la funci'on y en la segunda componente la derivada de la funci'on.
 
 Para cumplir este objetivo se pudiera definir directamente estas operaciones, pero usando
 como base la suma y la multiplicaci'on se puede recurrir a la teor'ia de 'algebras para obtener
 estas definiciones. La siguiente proposici'on brinda informaci'on sobre estas operaciones.
 
 \begin{proposicion}
  $\mathbb{A}$ es un anillo conmutativo y unitario. El cero es (0,0) y la identidad multiplicativa
  es (1,0).
 \end{proposicion}
 \begin{proof}
  Se obtiene directamente a partir de la definici'on del espacio $\mathbb{A}$
 \end{proof}
 
 El hecho de que $\mathbb{A}$ sea un anillo significa que todo elemento tiene un opuesto para la suma y 
 en este caso el opuesto de $(a,\dot{a})$ es $(-a,-\dot{a})$
 
 \begin{proposicion}
  Cuando se restan dos adoubles, en la primera componente se obtiene el resultado de la resta, y en la segunda,
  su derivada.
 \end{proposicion}
 \begin{proof}
  La resta es una operaci'on que se define como sumar con el opuesto, por lo
  que la resta de dos adoubles se define como:
  \begin{displaymath}
   (a,\dot{a})-(b,\dot{b})=(a-b,\dot{a}-\dot{b}), 	
  \end{displaymath}
  lo cual es el resultado que se quer'ia: la primera y segunda componente coinciden con la operaci'on y 
  derivada de la resta.
 \end{proof}
 
 La derivada del inverso es $\dot{a}/a^2$. Con la multiplicaci'on definida sobre $\mathbb{A}$ se puede demostrar
que en la segunda componente del inverso de un adouble se obtiene
precisamente la derivada del inverso. La siguiente proposici'on demuestra esto.
 
 \begin{proposicion}
  El inverso de un adouble $(a,\dot{a})$ es $(\frac{1}{a},-\frac{\dot{a}}{a^2})$.
 \end{proposicion}
 \begin{proof}
  Hallar el inverso de $(a,\dot{a})$ es equivalente a resolver la ecuaci'on
 \begin{displaymath}
  (a,\dot{a})(x,\dot{x})=(1,0),	
 \end{displaymath}
 donde $x$ y $\dot{x}$ son inc'ognitas. Esta ecuaci'on es equivalente al  
 sistema de ecuaciones
 \begin{displaymath}
  \left\{\begin{array}{l}
          ax=1\\
          a\dot{x}+\dot{a}x=0
         \end{array}\right.
 \end{displaymath}
 que tiene solución
 \begin{displaymath}
  x=\frac{1}{a}\qquad \textrm{y} \qquad \dot{x}=-\frac{\dot{a}}{a^2}.
 \end{displaymath}
 \end{proof}
 
 De la proposici'on anterior se deduce que en el caso de que $a=0$, los adoubles no tienen inverso.
 
 La divisi'on entre dos elementos de un anillo se realiza multiplicando uno por el inverso del otro. La siguiente
 proposici'on proporciona este resultado. 
 
 \begin{proposicion}
  El resultado de dividir $(a,\dot{a})$ por $(b,\dot{b})$ es 
  \begin{displaymath}
   \Big(\frac{a}{b},\frac{b\dot{a}-a\dot{b}}{b^2}\Big).
  \end{displaymath}
 \end{proposicion}
 \begin{proof}
  Para dividir
  \begin{displaymath}
	 (a,\dot{a})/(b,\dot{b}),
  \end{displaymath}
  se calcula
  \begin{displaymath}
	 (a,\dot{a})\Big(\frac{1}{b},-\frac{\dot{b}}{b^2}\Big) = 
	 \Big(\frac{a}{b},\frac{b\dot{a}-a\dot{b}}{b^2}\Big).
  \end{displaymath}
 \end{proof}
 
 Al no tener sentido hallar el inverso para $b = 0$, tampoco lo tendr'a la divisi'on cuando se quiera
 dividir por un adouble que tenga la primera componente igual a cero. Se puede notar que la primera componente
 de la divisi'on contiene el resultado del cociente y la segunda componente la regla de la derivada para
 la divisi'on.
 
 Hasta el momento se han definido las operaciones suma, resta,
 multiplicaci'on y divisi'on entre adoubles, y en todos los casos, en
 la segunda componente se obtiene la derivada de la operaci'on. Para
 poder utilizar los n'umeros adoubles para realizar AD falta definir
 funciones sobre $\mathbb{A}$ con el objetivo de obtener la su
 evaluaci'on y su derivada en la primera y segunda componente,
 respectivamente. En la siguiente secci'on se definen estas funciones
 de variable adouble.
 
 \section{Funciones de variable adouble} \label{sec:ad3}
 
 La definici'on de funciones se facilita si se introduce el cambio de
 notaci'on que se realiza en los n'umeros duales \cite{cheng}. Para
 poder utilizar esta nueva notaci'on es necesario definir sobre este
 espacio una estructura de 'algebra.
 
 Para definir un 'algebra sobre $\mathbb{A}$ conviene definir una
 estructura vectorial y despu'es demostrar su compatibilidad con el
 anillo. La siguiente proposici'on permite que la definici'on de
 espacio vectorial resulte natural.

 \begin{proposicion} \label{pro:contieneR}
  $\mathbb{A}$ contiene a $\mathbb{R}$ como estructura de anillo.
 \end{proposicion}
 \begin{proof}
  Tomando el conjunto de adoubles $R=\{(a,0)\in\mathbb{A}:a\in\mathbb{R}\}$, se tiene que 
  \begin{displaymath}
	 (a_1,0)-(a_2,0)=(a_1-a_2,0)\in R,
  \end{displaymath}
  \begin{displaymath}
	 (a_1,0)(a_2,0)=(a_1a_2,0)\in R,
  \end{displaymath}
  por lo que $R$ constituye un subanillo de $\mathbb{A}$, que es isomorfo a $\mathbb{R}$ si 
  se toma la aplicación que a cada adouble $(a,0)$ le hace corresponder el número real $a$.
 \end{proof}
 
 Una vez demostrado que $\mathbb{R}$ es un subanillo de $\mathbb{A}$ se puede definir
 un espacio vectorial sobre $\mathbb{A}$.
 
 \begin{proposicion}
  Sea $\alpha\in\mathbb{R}$ y $(a,\dot{a})\in\mathbb{A}$, si se define el producto exterior $\alpha(a,\dot{a})$   
  como
  \begin{displaymath}
	 \alpha(a,\dot{a}) = (\alpha,0)(a,\dot{a}) = (\alpha a,\alpha\dot{a}),
  \end{displaymath}
  se tiene que para este producto exterior y la suma definida en $\mathbb{A}$, este espacio es un espacio  
  vectorial.
 \end{proposicion}
 \begin{proof}
  Las operaciones de suma y producto externo coinciden con las usuales de $\mathbb{R}^2$ como espacio vectorial.
 \end{proof}
                    
 Como se tiene sobre $\mathbb{A}$ una estructura de anillo y de espacio vectorial se puede llegar a la
 siguiente proposici'on.
 
 \begin{proposicion}
  $\mathbb{A}$ es un álgebra conmutativa y unitaria.
 \end{proposicion}
 \begin{proof}
   Para realizar su demostraci'on s'olo basta
   demostrar que
   $$(\alpha\hat{x})(\beta\hat{y})=(\alpha\beta)\hat{x}\hat{y},$$ lo
   cual se obtiene de manera directa realizando las operaciones.
 \end{proof}

 En la siguiente subsecci'on se muestra un cambio de notaci'on que
 resulta de utilidad para realizar c'alculos de multiplicaci'on entre
 los adoubles, lo cual facilita la definici'on de funciones y el
 cálculo del cociente entre adoubles.

 \subsection{Unidad dual}

 El hecho de que $\mathbb{A}$ sea un 'algebra y que contenga a los reales permite hacer un cambio de
 notaci'on. Si se denota $\epsilon=(0,1)$, entonces si $\dot{a}\in\mathbb{R}$ se tiene que $(0,\dot{a})=\epsilon\dot{a}$, 
 por lo que cualquier elemento $(a,\dot{a})\in\mathbb{A}$ se puede expresar como 
 \begin{equation}
  (a,\dot{a})=(a,0)+(0,\dot{a})=(a,0)+\epsilon\dot{a}.	\label{nuevanotacion}
 \end{equation}
 Como el conjunto $R=\{(a,0)\in\mathbb{A}:a\in\mathbb{R}\}$ es isomorfo a $\mathbb{R}$, entonces se puede
 reescribir (\ref{nuevanotacion}) como 
 \begin{displaymath}
  (a,\dot{a})=a+\epsilon\dot{a}.
 \end{displaymath}

 Es usual llamar a $\epsilon$ unidad dual. La suma y la multiplicaci'on de dos adoubles usando la unidad dual resulta
 \begin{displaymath}
  (a+\epsilon \dot{a})+(b+\epsilon \dot{b})=a+b+\epsilon a+\epsilon b,
 \end{displaymath}
 \begin{displaymath}
	(a+\epsilon \dot{a})(b+\epsilon \dot{b})=ab+\epsilon a\dot{b}+\epsilon \dot{a}b+\epsilon^2\dot{a}\dot{b},
 \end{displaymath}
 pero como $\mathbb{A}$ es un 'algebra se puede sacar factor com'un $\epsilon$ y se obtiene
 \begin{displaymath}
  (a+\epsilon \dot{a})+(b+\epsilon \dot{b})=a+b+\epsilon(a+b),
 \end{displaymath}
 \begin{displaymath}
	(a+\epsilon \dot{a})(b+\epsilon \dot{b})=ab+\epsilon( a\dot{b}+\dot{a}b)+\epsilon^2\dot{a}\dot{b}.
 \end{displaymath}
 
 Para obtener el resultado final de la multiplicaci'on es necesario
 saber cómo opera la multiplicaci'on de $\epsilon$ consigo mismo.  La
 siguiente proposici'on muestra este resultado.

 \begin{proposicion} \label{pro:propiedadepsilon}
  $\epsilon^n=0, \forall n\geq 2$.
 \end{proposicion}
 \begin{proof}
  Como $\epsilon^2=(0,1)(0,1)=(0,0)=0$, entonces $$\epsilon^n=0, \forall n\geq 2$$.
 \end{proof}

 Con esta proposici'on se obtiene que la multiplicaci'on es
 \begin{displaymath}
	(a+\epsilon \dot{a})(b+\epsilon \dot{b})=ab+\epsilon( a\dot{b}+\dot{a}b).
 \end{displaymath}

 Como se puede notar en la suma y la multiplicaci'on usando la unidad dual se obtiene en la primera
 componente el resultado de la operaci'on y en la unidad dual la derivada de la operaci'on.
 
 \subsection{Definici'on de funciones}
  
 En esta subsecci'on se definen funciones sobre $\mathbb{A}$ a partir funciones reales. 
 Estas funciones tienen la caracter'istica que en la primera componente del resultado 
 de su evaluaci'on se obtiene el valor de la funci'on y en la segunda, la derivada. 
 Adem'as, estas funciones tendr'an la propiedad de ser extensiones de las funciones reales, en el 
 sentido de que si se eval'uan en un real se obtiene un valor real. 
 La siguiente proposici'on brinda una extensi'on de funciones reales con la propiedad que se desea.
 
 \begin{proposicion}
  Sea $f(x)$ una función real de variable real que puede ser representada a trav'es de su serie de 
  Taylor: \ $\displaystyle f(x)=\sum_{n=0}^{\infty}\frac{c_nx^n}{n!}$. Si $x=a+\epsilon\dot{a}$, entonces se
  cumple que: 
  \begin{displaymath}
	 f(a+\epsilon \dot{a}) = f(a)+\epsilon\dot{a}f'(a).
  \end{displaymath}
 \end{proposicion}
 \begin{proof}
  Como la fórmula del binomio
  \begin{displaymath}
   (x+\epsilon \dot{x})^n=\sum_{i=0}^n\binom{n}{i}x^{n-i}(\epsilon \dot{x})^i	
  \end{displaymath}
  se cumple para los anillos conmutativos (Teorema \ref{teo:binomio}), 
  entonces se tiene que:
  \begin{displaymath}
	 (x+\epsilon \dot{x})^n=\sum_{i=0}^n\binom{n}{i}x^{n-i}(\epsilon \dot{x})^i.
  \end{displaymath}
  Como $\epsilon^p=0\ \ \forall p>1$ (Proposici'on \ref{pro:propiedadepsilon}), entonces
  \begin{displaymath}
   (a+\epsilon \dot{a})^n=a^n+\epsilon pa^{n-1}\dot{a}	
  \end{displaymath}
  y
  \begin{displaymath}
   f(a+\epsilon \dot{a})=\sum_{n=0}^{\infty}\frac{c_n(a+\epsilon \dot{a})^n}{n!}=
   \sum_{n=0}^{\infty}\frac{c_na^n}{n!}+\epsilon \dot{a}  \sum_{n=1}^{\infty}\frac{c_na^{n-1}}{(n-1)!}=
   f(a)+\epsilon \dot{a}f'(a).	
  \end{displaymath}
 \end{proof}
   
 La funci'on $f(a+\epsilon \dot{a})=f(a)+\epsilon \dot{a}f'(a)$ es una
 extensi'on de los reales: si se toma $b=0$ se obtiene que
 $f(a)=f(a)$. En otras palabras, si se restringe la funci'on a los
 reales, se obtiene un real.
 
 Por otra parte, una funci'on $f$ definida de esta manera cumple la propiedad que se desea: en la primera
 componente se obtiene la evaluaci'on de $f$ y en la segunda se tiene su derivada al aplicar la regla de 
 la cadena.

 De manera similar a como se defini'o para el caso de una variable, una función de $n$ variables adoubles 
 se define como
 \begin{displaymath}
  f(x_1,\ldots,x_n)=f(a_1,\ldots,a_n) + \epsilon(\dot{a}_1f'_{x_1}(a_1,\ldots,a_n)+
  \cdots+\dot{a}_nf'_{x_n}(a_1,\ldots,a_n)),	
 \end{displaymath}
 donde $x_i=a_i+\epsilon \dot{a}_i$ y $f'_{x_i}(a_1,\ldots,a_n)$ es la
 derivada parcial de $f$ con respecto a la variable $x_i$, evaluada en
 $(a_1,\ldots,a_n)$. Al igual que en el caso de una variable, si
 $\dot{a}_i = 1$ y $\dot{a}_j = 0\ \forall j\neq i$, entonces
 \begin{displaymath}
  f(x_1,\ldots,x_n)=f(a_1,\ldots,a_n) + \epsilon f'_{x_i}(a_1,\ldots,a_n),	
 \end{displaymath}
 y en la segunda componente se obtiene la derivada parcial de $f$ respecto a $x_i$.
 
 En este cap'itulo se ha realizado una representaci'on algebraica de
 los n'umeros adoubles, que son el fundamento de la diferenciaci'on
 autom'atica.  Utilizando resultados del 'algebra se demostr'o que
 utilizando n'umeros adoubles se pueden obtener las derivadas de
 cualquier funci'on an'alitica.  En el siguiente cap'itulo se
 presentan los n'umeros superadoubles desde una perspectiva
 algebraica, y se demuestra que es posible calcular derivadas anidadas
 utilizando la metodolog'ia propuesta en el Cap'itulo
 \ref{cha:metodologia}.
 
\chapter{Espacio SuperAdouble}\label{cha:superadouble}

En este cap'itulo se demuestra el Teorema \ref{teo:teorema}, que
permite utilizar los números superadoubles para calcular derivadas
anidadas. La demostraci'on se realiza por inducci'on en las
operaciones involucradas en el c'alculo de la funci'on anidada.

Al igual que en el Cap'itulo \ref{cha:adouble}, se demuestra que el
espacio formado por los números superadoubles es un anillo, un espacio
vectorial y un 'algebra; y se definen funciones de variable
superadouble. En todos los casos se demuestra que con las operaciones
definidas se obtienen las derivadas anidadas.
 
La estructura de este cap'itulo es la siguiente: en la Secci'on
\ref{sec:casobase} se presenta el caso base de la inducci'on que
permitir'a demostrar el Teorema \ref{teo:teorema}; en la Secci'on
\ref{sec:introSA} se definir'an representaciones algebraicas de la
clase SuperAdouble y se demostrar'a el Teorema \ref{teo:teorema} para
las operaciones elementales. Por 'ultimo, en la Secci'on
\ref{sec:funcionesSA} se definir'an funciones de variable SuperAdouble
con las que se pueden calcular derivadas anidadas.
 
 \section{Preliminares y caso base} \label{sec:casobase}
 
 En esta secci'on se presenta la demostraci'on que para un caso base
 se cumple que los SuperAdouble definidos en el Cap'itulo
 \ref{cha:metodologia} permiten calcular derivadas anidadas.
 
 La clase SuperAdouble tiene dos campos: valor y derivada, y cada uno
 de estos campos tienen a su vez dos campos, uno valor y uno
 derivada. Por lo tanto, una variable de tipo SuperAdouble tiene
 cuatro campos: valor.valor, valor.  derivada, derivada.valor,
 derivada.derivada, por lo que se pueden respresentar como pares
 ordenados de adoubles o como elementos de $R^4$.
 
 El hecho de utilizar los SuperAdouble para evaluar y derivar
 funciones del tipo $f(x)=\dot{g}(h(x))$ implica que existe una
 funci'on anidada, por lo que en ese contexto también existe una
 derivada original, una derivada anidada y una derivada compuesta. La
 correspondencia existente entre estas derivadas y los campos de las
 variables de tipo SuperAdouble es la siguiente, de acuerdo con el
 Teorema \ref{teo:teorema}: el campo valor.valor de un superadouble es
 el valor de la funci'on anidada; el campo valor.derivada, es la
 derivada original; el campo derivada.valor es la derivada anidada, y
 en el campo derivada.derivada se obtiene la derivada compuesta.
 
 La demostraci'on que en estos campos se obtienen estas derivadas se
 hará por inducci'on. Como caso base se tomar'a un superadouble $X$
 que sea el resultado de una operación \emph{push}, pues esta
 es la primera operación que se hace con un SuperAdouble.

 En el teorema siguiente se demuestra que para una variable $X$ que
 sea el resultado de una operación \emph{push} se cumple el teorema
 \ref{teo:teorema}. Sin embargo, no tiene sentido hablar de derivada
 anidada y derivada compuesta si no existe una función anidada, y este
 es precisamente el caso de un superadouble que sea el resultado de un
 \emph{push}. Para que tenga sentido hablar de derivada anidada y
 derivada compuesta en el superadouble $X=push(x)$, en el teorema se
 considerará que la función anidada $g(x)$ es la función identidad.
 
 \begin{teorema} \label{teo:casobase} Sea $f(x)=\dot{g}(h(x))$ una
   funci'on donde se conocen $h(x)$ y $g(y)=y$. Si $(a,a')$ es el par
   de n'umeros reales que representan el valor y la derivada de
   $\mathbf{a} = h(x_0)$, entonces al aplicar $X =
   push((a,a'))=((a,a'),(1,0))$, y $Y=g(X)$, se obtiene que:
  \begin{itemize}
 	 \item $Y$.valor.valor es $g(h(x_0))$. 
	 \item $Y$.valor.derivada es $g'(h(x_0))$.
	 \item $Y$.derivada.valor es $\dot{g}(h(x_0))$.
	 \item $Y$.derivada.derivada es $\dot{g}'(h(x_0))$.
  \end{itemize}
 \end{teorema}
 \begin{proof}
  Por las hip'otesis se tiene que $h(x_0)=a$ y $h'(x_0)=a'$, por lo que 
  $g(h(x_0))=a$ y $g'(h(x_0))=\frac{dg(h(x))}{dx}=h'(x_0)=a'$. De esta
  forma se tiene los dos primeros puntos de la tesis.
  
  Como $\dot{g}(x)\equiv1$, entonces $\dot{g}(h(x_0))=1$, por lo que
  se cumple el tercer punto y como la derivada de una constante es
  cero, entonces se cumple el 'ultimo punto porque $\dot{g}'(x)\equiv
  0$.
 \end{proof}

 Este caso base, se utilizar'a como hip'otesis de inducci'on para
 demostrar que las operaciones elementales y las funciones sobre clase
 SuperAdouble calcula la derivada anidada y compuesta.

 En la pr'oxima secci'on se definen las operaciones elementales con
 sus respectivas demostraciones de que calculan las derivadas
 anidadas.
 
 \section{Introducción a $\mathbb{SA}$} \label{sec:introSA}
 
 En la secci'on anterior se demostr'o el Teorema
 \ref{teo:casobase}. Utilizando este teorema como caso base, en esta
 secci'on se aplicar'a inducci'on en las operaciones elementales que
 se realizan en una funci'on para demostrar el Teorema
 \ref{teo:teorema}.
 
 La siguiente definición se usar'a para realizar una representaci'on
 algebraica de la clase SuperAdouble usando la notaci'on introducida
 en la secci'on anterior.
 
 \begin{definicion} \label{def:superadouble1}
  El espacio SuperAdouble $\mathbb{SA}$ es el conjunto de pares ordenados $(\hat{x},\hat{y})$ donde
  $\hat{x}$ y $\hat{y}$ son adoubles y la suma y multiplicación se definen de la siguiente
  forma:
  \begin{displaymath}
   (\hat{a},\hat{b})+(\hat{c},\hat{d})=(\hat{a}+\hat{c},\hat{b}+\hat{d}),	
  \end{displaymath}
  \begin{displaymath}
   (\hat{a},\hat{b})(\hat{c},\hat{d})=(\hat{a}\hat{c},\hat{a}\hat{d}+\hat{b}\hat{c}).	
  \end{displaymath}
 \end{definicion}
 
 Estas operaciones definidas sobre $\mathbb{SA}$ se realizan
 multiplicando y sumando elementos del espacio $\mathbb{A}$, por
 ejemplo, la primera componente de la suma es $\hat{a}+\hat{c}$, por
 tanto, si
 \begin{displaymath}
  \hat{a} = (a,\dot{a}),
 \end{displaymath}
 \begin{displaymath}
  \hat{c} = (c,\dot{c}),	
 \end{displaymath}
 entonces la primera componente de la suma es el adouble
 \begin{displaymath}
  \hat{a}+\hat{c} = (a + c, \dot{a} + \dot{c}).	
 \end{displaymath}
 
 A los elementos del espacio SuperAdouble se les llamar'a n'umeros superadoubles
 o simplemente superadoubles.
 
 Lo que se quiere probar es que los superadoubles calculan derivadas
 anidadas en su tercera y cuarta componente. Para esto resulta más
 conveniente considerar los superadoubles como elementos de
 $\mathbb{R}^4$, por lo que resultaría de mucha utilidad encontrar un
 isomorfismo entre $\mathbb{SA}$ y un subconjunto de
 $\mathbb{R}^4$. Para tener un isomorfismo con las operaciones
 definidas es necesario tener alguna estructura algebraica como la
 proposici'on siguiente.
 
 \begin{proposicion}
   \label{prop:sa-anillo-conmutativo-unitario}
  $\mathbb{SA}$ es un anillo conmutativo y unitario. El neutro con respecto a la 
  suma es el elemento $\mathbf{0} = ((0,0),(0,0))$ y el neutro respecto al producto es 
  $\mathbf{1} = ((1,0),(0,0))$. 
 \end{proposicion}
 \begin{proof}
  La demostraci'on es inmediata a partir de la definici'on del espacio, el lector interesado puede
  leerla en los anexos. 
 \end{proof}
 
 Si se tiene dos elementos de $\mathbb{SA}$ entonces con la proposici'on anterior la resta se obtiene 
 sumando por el opuesto del segundo elemento:
 \begin{displaymath}
  (\hat{a},\hat{b})-(\hat{c},\hat{d})=(\hat{a},\hat{b})+(-\hat{c},-\hat{d}) = (\hat{a}-\hat{c},\hat{b}-\hat{d}).
 \end{displaymath}
 
 A continuaci'on se presenta otra definici'on del espacio $\mathbb{SA}$ en el sentido de isomorfismo.
 Esta nueva definici'on es importante porque se realiza directamente sobre los reales y no sobre $\mathbb{A}$,  
 por lo que facilita el trabajo de c'alculo cuando se definan la divisi'on de superadoubles y las funciones
 de variable superadouble.

 \begin{definicion} \label{def:superadouble2}
  El espacio SuperAdouble \ $\mathbb{SA}_{\mathbb{R}}$ es el conjunto de $\mathbb{R}^4$ con la suma y multiplicación 
  definidos por
  {\setlength\arraycolsep{2pt}
   \begin{eqnarray}
    \mathbf{x}+\mathbf{y} & = & (x+y,x'+y',\dot{x}+\dot{y},\dot{x}'+\dot{y}'), \label{eq:suma} \\
    \mathbf{x}\mathbf{y} & = & (xy, \nonumber \\
                         &   &  xy'+x'y,\nonumber \\
                         &   &  x\dot{y}+\dot{x}y, \nonumber \\
                         &   &  x\dot{y}'+x'\dot{y}+\dot{x}y'+\dot{x}'y), \label{eq:producto}
   \end{eqnarray}
  }donde $\mathbf{x}=(x,x',\dot{x},\dot{x}')$ y $\mathbf{y}=(y,y',\dot{y},\dot{y}')$.
 \end{definicion}
 
 \begin{proposicion} \label{prop:SAR-anillo-conmutativo}
  El espacio $\mathbb{SA}_{\mathbb{R}}$ es un anillo conmutativo.
 \end{proposicion}
 \begin{proof}
  La demostraci'on se deduce directo de la definici'on de anillo, el lector interesado
  en la demostraci'on puede referirse a los anexos.
 \end{proof}
 
 A pesar de que $\mathbb{SA}_\mathbb{R}$ es un anillo unitario, no se
 demostr'o en la proposici'on anterior ya que con la siguiente
 proposici'on se obtiene como consecuencia del isomorfismo. Adem'as de
 esta propiedad, el isomorfismo entre $\mathbb{SA}$ y
 $\mathbb{SA}_\mathbb{R}$ brinda la posibilidad de usar cualquiera de
 estos dos para representar la clase SuperAdouble.
 
 \begin{proposicion} 
  Sea $\phi:\mathbb{SA}\longrightarrow\mathbb{SA}_\mathbb{R}$ la aplicaci'on
  \begin{displaymath}
	 \phi(((a,\dot{a}),(b,\dot{b}))) = (a,\dot{a},b,\dot{b}).
  \end{displaymath}
  Con esta aplicaci'on se tiene que $\mathbb{SA}$ es isomorfo a $\mathbb{SA}_\mathbb{R}$. 
 \end{proposicion}
 \begin{proof}
  Sea
  \begin{displaymath}
   \hat{a}=(a,\dot{a}), \qquad \hat{b}=(b,\dot{b}), \qquad \hat{c}=(c,\dot{c}), \qquad \hat{d}=(d,\dot{d}).	
  \end{displaymath}
  
  Como la aplicaci'on $\phi$ es una biyecci'on, s'olo falta demostrar
  \begin{displaymath}
	 \phi((\hat{a},\hat{b})+(\hat{c},\hat{d})) = \phi((\hat{a},\hat{b}))+\phi((\hat{c},\hat{d})),
  \end{displaymath}
  \begin{displaymath}
	 \phi((\hat{a},\hat{b})(\hat{c},\hat{d})) = \phi((\hat{a},\hat{b}))\phi((\hat{c},\hat{d})).
	\end{displaymath}
	
	Calculando
	{\setlength\arraycolsep{2pt}
   \begin{eqnarray*}
    \phi((\hat{a},\hat{b})+(\hat{c},\hat{d})) & = & \phi(((a+c,\dot{a}+\dot{c}),(b+d,\dot{b}+\dot{d}))) \\
	                                            & = & (a+c,\dot{a}+\dot{c},b+d,\dot{b}+\dot{d}),
   \end{eqnarray*}}
  {\setlength\arraycolsep{2pt}
   \begin{eqnarray*}
    \phi((\hat{a},\hat{b}))+\phi((\hat{c},\hat{d})) & = & (a,\dot{a},b,\dot{b})+(c,\dot{c},d,\dot{d}) \\
	                                                  & = & (a+c,\dot{a}+\dot{c},b+d,\dot{b}+\dot{d}), 
   \end{eqnarray*}
   }se obtiene que 
  \begin{displaymath}
	 \phi((\hat{a},\hat{b})+(\hat{c},\hat{d})) = \phi((\hat{a},\hat{b}))+\phi((\hat{c},\hat{d})).
  \end{displaymath}
  
  Con la multiplicaci'on se opera an'alogo a la suma, pero primero se 
  necesita realizar los c'alculos
  \begin{displaymath}
   \hat{a}\hat{c}=(ac,a\dot{c}+\dot{a}c),	
  \end{displaymath}
  \begin{displaymath}
   \hat{a}\hat{d}=(ad,a\dot{d}+\dot{a}d),	
  \end{displaymath}
  \begin{displaymath}
   \hat{b}\hat{c}=(bc,b\dot{c}+\dot{b}c),	
  \end{displaymath}
  que intervienen en la siguiente operaci'on
  \begin{displaymath}
	 \hat{a}\hat{d}+\hat{b}\hat{c}=(ad+bc,a\dot{d}+\dot{a}d+b\dot{c}+\dot{b}c).
  \end{displaymath}
  Usando la expresi'on anterior se calcula
  \begin{displaymath}
	 \phi((\hat{a},\hat{b})(\hat{c},\hat{d})) = \phi((\hat{a}\hat{c},\hat{a}\hat{d}+\hat{b}\hat{c})) 
	                                          = (ac,a\dot{c}+\dot{a}c,ad+bc,a\dot{d}+\dot{a}d+b\dot{c}+\dot{b}c).
  \end{displaymath}
  Por otra parte
  {\setlength\arraycolsep{2pt}
   \begin{eqnarray*}
    \phi((\hat{a},\hat{b}))\phi((\hat{c},\hat{d})) & = & (a,\dot{a},b,\dot{b})(c,\dot{c},d,\dot{d}) \\
	                                                 & = & (ac,a\dot{c}+\dot{a}c,ad+bc,a\dot{d}+\dot{a}d+b\dot{c}+\dot{b}c). 
   \end{eqnarray*}
   }por lo que se obtiene
  \begin{displaymath}
	 \phi((\hat{a},\hat{b})(\hat{c},\hat{d})) = \phi((\hat{a},\hat{b}))\phi((\hat{c},\hat{d})).
  \end{displaymath}
 \end{proof}

 En lo que sigue se usar'a la notaci'on $\dot{(a+b)}, \dot{(a-b)}$ y $\dot{(ab)}$ para representar
 la derivada anidada de la operaci'on correspondiente. 
 
 El siguiente teorema demuestra que los superadoubles
 calculan derivadas anidadas en las tres operaciones definidas hasta el momento.
 
 \begin{teorema} \label{teo:operaciones} Sean $(a,a',\dot{a},\dot{a}')$
   y $(b,b',\dot{b},\dot{b}')$ dos superadoubles cuyas componentes
   representan el valor, la derivada original, la derivada anidada y
   la derivada compuesta de sus respectivas operaciones. Entonces se
   cumple
  \begin{displaymath}
	 (a,a',\dot{a},\dot{a}')+(b,b',\dot{b},\dot{b}')=(a+b,(a+b)',\dot{(a+b)},(\dot{(a+b)})'),
  \end{displaymath}
  \begin{displaymath}
	 (a,a',\dot{a},\dot{a}')-(b,b',\dot{b},\dot{b}')(a-b,(a-b)',\dot{(a-b)},(\dot{(a-b)})'),
  \end{displaymath}
  \begin{displaymath}
	 (a,a',\dot{a},\dot{a}')(b,b',\dot{b},\dot{b}') = (ab,(ab)',\dot{(ab)},(\dot{(ab)})').
  \end{displaymath}
 \end{teorema}
 \begin{proof}
  Sumando y restando se obtiene
  \begin{displaymath}
	 (a,a',\dot{a},\dot{a}')+(b,b',\dot{b},\dot{b}') = (a+b,a'+b',\dot{a}+\dot{b},\dot{a}'+\dot{b}'),
  \end{displaymath}
  \begin{displaymath}
	 (a,a',\dot{a},\dot{a}')-(b,b',\dot{b},\dot{b}') = (a-b,a'-b',\dot{a}-\dot{b},\dot{a}'-\dot{b}').
  \end{displaymath}
  Usando que la suma de la derivada es la derivada de la suma se obtiene
  {\setlength\arraycolsep{2pt}
   \begin{eqnarray*}
    (a+b,a'+b',\dot{a}+\dot{b},\dot{a}'+\dot{b}') & = & (a+b,(a+b)',\dot{(a+b)},(\dot{a}+\dot{b})') \\
	  									                            & = & (a+b,(a+b)',\dot{(a+b)},(\dot{(a+b)})'), 
   \end{eqnarray*}}
  {\setlength\arraycolsep{2pt}
   \begin{eqnarray*}
    (a-b,a'-b',\dot{a}-\dot{b},\dot{a}'-\dot{b}') & = & (a-b,(a-b)',\dot{(a-b)},(\dot{a}-\dot{b})') \\
	  									                            & = & (a-b,(a-b)',\dot{(a-b)},(\dot{(a-b)})'), 
   \end{eqnarray*}} 
  Para demostrar el tercer resultado de la tesis se aplica la regla del producto para
  las derivadas,
  {\setlength\arraycolsep{2pt}
   \begin{eqnarray*}
    (a,a',\dot{a},\dot{a}')(b,b',\dot{b},\dot{b}') & = & (ab,a'b+ab',\dot{a}b+a\dot{b},(a\dot{b}'+a'\dot{b})+(\dot{a}'b+\dot{a}b'))\\
                                                   & = & (ab,(ab)',\dot{(ab)},((a\dot{b})'+\dot{a}b)')\\
                                                   & = & (ab,(ab)',\dot{(ab)},(\dot{a}b+(a\dot{b})')\\
                                                   & = & (ab,(ab)',\dot{(ab)},(\dot{(ab)})').
   \end{eqnarray*}}
 \end{proof}
 
 La divisi'on entre dos elementos se calcula multiplicando uno por el inverso
 del otro, por lo tanto, para realizar la divisi'on primero se necesita el inverso.

 \begin{proposicion}
  El inverso de $\mathbf{a}=(a,a',\dot{a},\dot{a}')$ es
  $$\mathbf{a}^{-1}=
  \Big(\frac{1}{a},-\frac{a'}{a^2},-\frac{\dot{a}}{a^2},\frac{2a'\dot{a}-a\dot{a}'}{a^3}\Big).$$
 \end{proposicion}
 \begin{proof}
   El inverso del superadouble $\mathbf{a}$ es
   $\mathbf{a}^{-1}=(x,x',\dot{x}, \dot{x}')$ tal que 
  \begin{displaymath}
   (a,a',\dot{a},\dot{a}')(x,x',\dot{x},\dot{x}')=(1,0,0,0).	
  \end{displaymath}
  Hallar $\mathbf{a}^{-1}$ es equivalente a resolver el sistema de
  ecuaciones
  \begin{displaymath}
   \left\{\begin{array}{l}
          ax=1\\
          ax'+a'x=0\\
          a\dot{x}+\dot{a}x=0\\
          a\dot{x}'+a'\dot{x}+\dot{a}x'+\dot{a}'x=0
         \end{array}
   \right.	
  \end{displaymath}
  cuya solución es:
  \begin{displaymath}
   x=\frac{1}{a},	
  \end{displaymath}
  \begin{displaymath}
   x'=-\frac{a'}{a^2},	
  \end{displaymath}
  \begin{displaymath}
   \dot{x}=-\frac{\dot{a}}{a^2},	
  \end{displaymath}
  \begin{displaymath}
   \dot{x}'=\frac{2a'\dot{a}-a\dot{a}'}{a^3}.	
  \end{displaymath}
 \end{proof}
 
 La soluci'on anterior no tiene sentido para $a=0$, esto significa que
 los elementos con primera componente igual a cero no son inversibles.
 
 El siguiente teorema demuestra que el c'alculo de la derivada anidada
 del inverso es posible.
 
 \begin{teorema} \label{teo:inverso} Sea $\mathbf{a} =
   (a,a',\dot{a},\dot{a}')$ un superadouble cuyas componentes son el
   valor, la derivada original, la derivada anidada y la derivada
   compuesta de la operación que él representa. Entonces se cumple
   que:
  \begin{displaymath}
	 \mathbf{a}^{-1} = \Big(\frac{1}{a},\Big(\frac{1}{a}\Big)',\dot{\Big(\frac{1}{a}\Big)},\dot{\Big(\frac{1}{a}\Big)}'\Big).
  \end{displaymath}
 \end{teorema}
 \begin{proof}
  Utilizando la regla de la derivada para la inversa y el producto se obtiene
  {\setlength\arraycolsep{2pt}
   \begin{eqnarray*}
    \mathbf{a}^{-1} & = & \Big(\frac{1}{a},-\frac{a'}{a^2},-\frac{\dot{a}}{a^2},\frac{2a'\dot{a}-a\dot{a}'}{a^3}\Big) \\
	                  & = & \Big(\frac{1}{a},\Big(\frac{1}{a}\Big)',\dot{\Big(\frac{1}{a}\Big)},\frac{2aa'\dot{a}}{a^4}-\frac{\dot{a}'}{a^2}\Big) \\
	                  & = & \Big(\frac{1}{a},\Big(\frac{1}{a}\Big)',\dot{\Big(\frac{1}{a}\Big)},\Big(-\dot{a}\cdot \frac{1}{a^2}\Big)'\Big) \\
	                  & = & \Big(\frac{1}{a},\Big(\frac{1}{a}\Big)',\dot{\Big(\frac{1}{a}\Big)},\dot{\Big(\frac{1}{a}\Big)}'\Big). 
   \end{eqnarray*}}
 \end{proof}
 
 \begin{corolario}
   Sea $(a,a',\dot{a},\dot{a}')$ y $(b,b',\dot{b},\dot{b}')$ dos
   superadoubles cuyas componentes son el valor, la derivada
   original, la derivada anidada y la derivada compuesta de las
   operaciones que ellas representan. Entonces se
   cumple
  \begin{displaymath}
    (a,a',\dot{a},\dot{a}')/(b,b',\dot{b},\dot{b}')=(a/b,(a/b)',\dot{(a/b)},(\dot{a/b})'),
  \end{displaymath}
 \end{corolario}
 \begin{proof}
   Como la divisi'on de dos elementos es la multiplicaci'on de uno por
   el inverso del otro, entonces por el Teorema \ref{teo:inverso} y el
   Teorema \ref{teo:operaciones} se cumple la tesis.
 \end{proof} 
 
 Hasta este momento se ha demostrado que las operaciones elementales
 sobre $\mathbb{SA}$ permiten calcular derivadas anidadas. Para poder
 utilizar los superadoubles con cualquier funci'on diferenciable,
 falta probar que estas propiedades tambi'en se cumplen cuando se
 eval'ua una funci'on anal'itica cualquiera. El objetivo de la
 pr'oxima secci'on es precisamente la definici'on de funciones de
 variable superadouble.
 
 \section{Funciones de variable superadouble} \label{sec:funcionesSA}
 
 Al igual que los adoubles, las definiciones de funciones se facilita 
 si se introduce un cambio de notaci'on en los superadoubles. Para esta nueva notaci'on 
 es necesario definir sobre este espacio una estructura de 'algebra y para definir
 un 'algebra sobre $\mathbb{SA}$ falta definir una estructura vectorial y despu'es demostrar su
 compatibilidad con el anillo.
 
 La siguiente proposici'on ayuda a que la definici'on de espacio vectorial
 resulte natural. 
  
 \begin{proposicion}
  $\mathbb{SA}$ contiene a $\mathbb{A}$ y $\mathbb{R}$ como estructura de anillo.
 \end{proposicion}
 \begin{proof}
  Tomando $A = \{(a,a',0,0)\in\mathbb{SA}:a,a'\in \mathbb{R}\}$ y restando 
  y multiplicando dos elementos cualesquiera de este conjunto se obtiene que 
  {\setlength\arraycolsep{2pt}
   \begin{eqnarray*}
    (a,a',0,0)-(b,b',0,0)&=&(a-b,a'-b',0,0)\in A, \\
    (a,a',0,0)(b,b',0,0)&=&(ab,ab'+a'b,0,0)\in A,
   \end{eqnarray*}
  }por tanto, $A$ es un subanillo de $\mathbb{SA}$ que es isomorfo a $\mathbb{A}$. 
  
  Como $\mathbb{SA}$ contiene a $\mathbb{A}$, esto implica que $\mathbb{SA}$ también 
  contiene a los reales por la Proposici'on \ref{pro:contieneR}. Este subanillo es
  el conjunto $\{(a,0,0,0)\in\mathbb{SA}:a\in\mathbb{R}\}$.
 \end{proof}
 
 \begin{proposicion}
  El espacio SuperAdouble en un espacio vectorial sobre $\mathbb{R}$ con la 
  suma $\eqref{eq:suma}$ y el producto externo definido por
  \begin{equation}
   \alpha(a,a',\dot{a},\dot{a}')=(\alpha,0,0,0)(a,a',\dot{a},\dot{a}')=(\alpha a,\alpha a',\alpha \dot{a},\alpha \dot{a}') \label{eq:productoR}.
  \end{equation}
 \end{proposicion}
 \begin{proof}
  Las operaciones de suma y producto externo coinciden con las usuales sobre $\mathbb{R}^4$ como espacio vectorial.
 \end{proof}

 \begin{proposicion}
  $\mathbb{SA}$ es un álgebra conmutativa y unitaria con las operaciones $\eqref{eq:suma}, 
  \eqref{eq:producto}$ y \eqref{eq:productoR}.
 \end{proposicion}
 \begin{proof}
  Hay que demostrar que $\forall\alpha,\beta\in\mathbb{R}$ y $\forall\mathbf{a},\mathbf{b}\in\mathbb{SA}$ se cumple 
  \begin{displaymath}
	 (\alpha\mathbf{a})(\beta\mathbf{b})=(\alpha\beta)\mathbf{a}\mathbf{b}.
  \end{displaymath}
  
  Sea $\mathbf{a}=(a,a',\dot{a},\dot{a}')$ y $\mathbf{b}=(b,b',\dot{b},\dot{b}')$, 
  entonces se cumple 
  {\setlength\arraycolsep{2pt}
   \begin{eqnarray*}
    (\alpha\mathbf{a})(\beta\mathbf{b}) & = & (\alpha\beta ab,\alpha\beta (ab'+a'b),
                                               \alpha\beta (a\dot{b}+b\dot{a}), \\
                                        &   &  \alpha\beta (a\dot{b}'+a'\dot{b}+\dot{a}b'+\dot{a}'b)) \\
                                        & = & (\alpha\beta)\mathbf{a}\mathbf{b}.
   \end{eqnarray*}}
  \end{proof}

 Despu'es de esta proposici'on se est'a en condiciones de presentar una nueva notaci'on con el prop'osito comentado
 anteriormente: facilitar la definici'on de funciones de variable superadouble.
 
 Como 
 \begin{displaymath}
  (a,a',\dot{a},\dot{a}') = (a,0,0,0)+(0,a',0,0)+(0,0,\dot{a},0)+(0,0,0,\dot{a}'),	
 \end{displaymath}
 se puede denotar
 \begin{displaymath}
  \epsilon_1= (0,1,0,0),	
 \end{displaymath}
 \begin{displaymath}
  \epsilon_2=(0,0,1,0),	
 \end{displaymath}
 \begin{displaymath}
  \epsilon_3=(0,0,0,1),	
 \end{displaymath}
 y se podría reescribir los superadoubles como
 \begin{displaymath}
  (a,a',\dot{a},\dot{a}') = a+\epsilon_1 a'+\epsilon_2 \dot{a}+\epsilon_3 \dot{a}'.	
 \end{displaymath}
 
 Como la definici'on de funciones se realizar'a a trav'es de la serie
 de Taylor de las funciones reales, se encontrar'an factores donde
 aparecen los productos $\epsilon_i\epsilon_j$. La siguiente
 proposici'on proporciona estos resultados.
 
 \begin{proposicion} \label{pro:epsilons}
  Los productos de cada par de valores $\epsilon_i,\epsilon_j$ son
  \begin{displaymath}
	 \epsilon_1\epsilon_1=0,
	\end{displaymath}
	\begin{displaymath}
   \epsilon_1\epsilon_2=\epsilon_3,
  \end{displaymath}
  \begin{displaymath}
   \epsilon_1\epsilon_3=0,
  \end{displaymath}
  \begin{displaymath}
   \epsilon_2\epsilon_2=0,
  \end{displaymath}
  \begin{displaymath}
   \epsilon_2\epsilon_3=0,
  \end{displaymath} 
  \begin{displaymath}
   \epsilon_3\epsilon_3=0.
  \end{displaymath}
 \end{proposicion}
 \begin{proof}
  Estos resultados se obtienen de realizar los siguientes
  c'alculos.
  {\setlength\arraycolsep{2pt}
  \begin{eqnarray*}
   \epsilon_1\epsilon_1&=&(0,1,0,0)(0,1,0,0)=0,\\
   \epsilon_1\epsilon_2&=&(0,1,0,0)(0,0,1,0)=\epsilon_3,\\
   \epsilon_1\epsilon_3&=&(0,1,0,0)(0,0,0,1)=0,\\
   \epsilon_2\epsilon_2&=&(0,0,1,0)(0,0,1,0)=0,\\
   \epsilon_2\epsilon_3&=&(0,0,1,0)(0,0,0,1)=0,\\
   \epsilon_3\epsilon_3&=&(0,0,0,1)(0,0,0,1)=0.
  \end{eqnarray*}
 }Solo es necesario hacer estos seis cálculos porque el álgebra es conmutativa.
 \end{proof}
 
 Estos productos se resumen en la siguiente tabla.
 \begin{displaymath}
	\begin{tabular}{|c|c|c|c|}
   \hline
   &$\epsilon_1$&$\epsilon_2$&$\epsilon_3$ \\
   \hline
   $\epsilon_1$&0&$\epsilon_3$&0 \\
   \hline
   $\epsilon_2$&$\epsilon_3$&0&0 \\
   \hline
   $\epsilon_3$&0&0&0\\
   \hline
  \end{tabular}
 \end{displaymath}
 
 Lo que se quiere en este punto es definir sobre $\mathbb{SA}$
 funciones que permitan calcular derivadas anidadas y que sean
 compatible con $\mathbb{A}$ y $\mathbb{R}$, es decir, que al
 restringir una función de $\mathbb{SA}$ a alguno de estos dos
 espacios, se obtengan valores de estos espacios. Para definir estas
 funciones, al igual que con los adoubles, se partir'a del desarrollo
 en serie de Taylor de las funciones reales anal'iticas.

 \begin{teorema}
   Sea $f$ una función real anal'itica, $y=h(x)$ la variable anidada y
   $x$ la variable original. Entonces si
   $\mathbf{a}=a+\epsilon_1a'+\epsilon_2\dot{a}+\epsilon_3\dot{a}'$, y
   $\mathbf{b} =
   f(\mathbf{a})=b+\epsilon_1b'+\epsilon_2\dot{b}+\epsilon_3\dot{b}'$,
   se tiene que
   \begin{enumerate}
     \item $b$ es el valor de la función: f(a),
     \item $b'$ es la derivada original: $\frac{df(a)}{dx}$,
     \item $\dot{b}$ es la derivada anidada: $\frac{df(a)}{dy}$, y
     \item $\dot{b}'$ es la derivada compuesta: $\frac{d^2f(a)}{dxdy}$.
   \end{enumerate}
 \end{teorema}
 \begin{proof}
  Como $\mathbb{SA}$ es un 'algebra conmutativa, por el Teorema \ref{teo:binomio}
  se cumple
  \begin{equation}
   (a+\epsilon_1 a'+\epsilon_2 \dot{a}+\epsilon_3 \dot{a}')^n=
   \sum{\frac{n!}{p_1!p_2!p_3!p_4!}a^{p_1}(\epsilon_1 a')^{p_2}
   (\epsilon_2 \dot{a})^{p_3}(\epsilon_3\dot{a}')^{p_4}}, \label{eq:cuatrinomio}
  \end{equation}
  donde la suma del segundo miembro se extiende a todas las combinaciones de 
  $p_1, p_2, p_3$ y $p_4$ de enteros no negativos tales que
  \begin{displaymath}
   \sum_{i=1}^4p_i=n.	
  \end{displaymath}

 En la suma \eqref{eq:cuatrinomio} existen muchos factores que se anulan
 cuando se aplica la Proposici'on \ref{pro:epsilons}, por lo tanto, 
 s'olo se tratarán los cinco casos donde esto no ocurre.\\ \\
 \textbf{Caso 1:} Si $p_2=p_3=p_4=0$, entonces $p_1=n$ y 
 \[\frac{n!}{p_1!p_2!p_3!p_4!}a^{p_1}(\epsilon_1 a')^{p_2}(\epsilon_2 \dot{a})^{p_3}(\epsilon_3\dot{a}')^{p_4}=a^n.\]
 \textbf{Caso 2:} Si $p_2=1$ y $p_3=p_4=0$, entonces $p_1=n-1$ y 
 \[\frac{n!}{p_1!p_2!p_3!p_4!}a^{p_1}(\epsilon_1 a')^{p_2}(\epsilon_2 \dot{a})^{p_3}(\epsilon_3\dot{a}')^{p_4}=\epsilon_1na^{n-1}a'.\]
 \textbf{Caso 3:} Si $p_3=1$ y $p_2=p_4=0$, entonces $p_1=n-1$ y 
 \[\frac{n!}{p_1!p_2!p_3!p_4!}a^{p_1}(\epsilon_1 a')^{p_2}(\epsilon_2 \dot{a})^{p_3}(\epsilon_3\dot{a}')^{p_4}=\epsilon_2na^{n-1}\dot{a}.\]
 \textbf{Caso 4:} Si $p_4=1$ y $p_2=p_3=0$, entonces $p_1=n-1$ y 
 \[\frac{n!}{p_1!p_2!p_3!p_4!}a^{p_1}(\epsilon_1 a')^{p_2}(\epsilon_2 \dot{a})^{p_3}(\epsilon_3\dot{a}')^{p_4}=\epsilon_3na^{n-1}\dot{a}'.\]
 \textbf{Caso 5:} Si $p_2=p_3=1$ y $p_4=0$, entonces $p_1=n-2$ y 
 \[\frac{n!}{p_1!p_2!p_3!p_4!}a^{p_1}(\epsilon_1 a')^{p_2}(\epsilon_2 \dot{a})^{p_3}(\epsilon_3\dot{a}')^{p_4}=\epsilon_3(n-1)na^{n-2}a'\dot{a}.\]

 Considerando los resultados obtenidos en los cinco casos a partir de (\ref{eq:cuatrinomio}) 
 y us'andolos en (\ref{eq:superfuncion}) se obtiene que
 {\setlength\arraycolsep{2pt}
  \begin{eqnarray*}
   f((a,a',\dot{a},\dot{a}'))&=&\sum_{n=0}^\infty \frac{c_n(a+\epsilon_1a'+\epsilon_2\dot{a}+\epsilon_3\dot{a}')^n}{n!} \\                   &=&\sum_{n=0}^\infty\frac{c_na^n}{n!}+\epsilon_1a'\sum_{n=1}^\infty \frac{c_na^{n-1}}{(n-1)!}+\epsilon_2\dot{a}\sum_{n=1}^\infty \frac{c_na^{n-1}}{(n-1)!} \\  & & \ + \epsilon_3\Bigg(a'\dot{a}\sum_{n=2}^\infty \frac{c_na^{n-2}}{(n-2)!} + \dot{a}'\sum_{n=1}^\infty \frac{c_na^{n-1}}{(n-1)!}\Bigg).
  \end{eqnarray*}
 }

 Como
 \begin{displaymath}
  \sum_{n=1}^\infty \frac{c_na^{n-1}}{(n-1)!}=\dot{f}(a),	 
 \end{displaymath}
 \begin{displaymath}
 	\sum_{n=2}^\infty \frac{c_na^{n-2}}{(n-2)!}=\ddot{f}(a),
 \end{displaymath}
 se concluye que si
 $\mathbf{a}=a+\epsilon_1a'+\epsilon_2\dot{a}+\epsilon_3\dot{a}'$,
 entonces
  \begin{equation}  \label{eq:superfuncion}
	 f(\mathbf{a})=f(a)+\epsilon_1a'\dot{f}(a)+\epsilon_2\dot{a}\dot{f}(a)+
   \epsilon_3(a'\dot{a}\ddot{f}(a)+\dot{a}'\dot{f}(a)). 
  \end{equation}
  
  Usando la regla de la cadena se obtiene 
  \begin{displaymath}
	 f(\mathbf{a})=b+ \epsilon_1 b' + \epsilon_2 \dot{b} + \epsilon_3 \dot{b}'.
  \end{displaymath}
  
  Para obtener el resultado en la cuarta componente se aplica la regla de
  la cadena y la regla del producto para las derivadas como sigue:
  \begin{displaymath}
	 a'\dot{a}\ddot{f}(a)+\dot{a}'\dot{f}(a)=(\dot{a}\dot{f}(a))'=\dot{f}'(a). 
  \end{displaymath}
 \end{proof}
 
 Para comprobar que \eqref{eq:superfuncion} es una extensi'on de $\mathbb{R}$ a $\mathbb{SA}$, 
 basta evaluarlo en puntos de la forma $(a,0,0,0)$ y para ver que es una extensi'on de
 $\mathbb{A}$ a $\mathbb{SA}$, s'olo es necesario evaluarlo en puntos de la forma
 $(a,a',0,0)$. 

 Análogo al espacio Adouble, se puede definir una función de $n$
 variables superadoubles de la siguiente forma:
 {\setlength\arraycolsep{2pt}
  \begin{eqnarray*}
   f(x_1,\ldots,x_n) & = & f(a_{11},\ldots,a_{1n}) \\ 
   & & +\ \epsilon_1(a_{21}f'_{x_1}(a_{11},\ldots,a_{1n}) +\cdots+a_{2n}f'_{x_n}(a_{11},\ldots,a_{1n})) \\
   & & +\ \epsilon_2(b_{11}f'_{x_1}(a_{11},\ldots,a_{1n}) +\cdots+b_{1n}f'_{x_n}(a_{11},\ldots,a_{1n})) \\
   & & +\ \epsilon_3(a_{21}b_{11}f''_{x_1}(a_{11},\ldots,a_{1n}) +\cdots+a_{2n}b_{1n}f''_{x_n}(a_{11},\ldots,a_{1n}) \\
   & & +\ b_{21}f'_{x_1}(a_{11},\ldots,a_{1n}) +\cdots+b_{2n}f'_{x_n}(a_{11},\ldots,a_{1n})),
  \end{eqnarray*}
 }donde $x_i=(a_{1i},a_{2i},b_{1i},b_{2i})$ para todo $i=1,\ldots,n$.

 Todo lo desarrollado en esta sección para funciones de una variable
 superadouble se cumple para estas funciones en varias variables.

 En este cap'itulo se demostr'o que la metodolog'ia presentada en el
 Cap'itulo \ref{cha:metodologia} es una herramienta v'alida para
 calcular derivadas anidadas.

\chapter*{Conclusiones}
 
 \addcontentsline{toc}{chapter}{Conclusiones}

 Con la metodología propuesta en este trabajo es posible utilizar la
 diferenciación automática para calcular derivadas anidadas. Esta
 metodología es sencilla de implementar gracias a que se puede reusar
 otras librerías de AD que soporten sobrecarga de operadores.  Además,
 se hizo un estudio algebraico de las propiedades de los números que
 sustentan tanto la diferenciación automática como la diferenciación
 automática anidada.

 Como recomendaciones y trabajo futuro se propone implementar el modo
 hacia atrás de la diferenciación automática con los números de tipo
 SuperAdouble, lo cual sería útil en los casos que la función anidada
 con dominio en $\mathbb{R}^n$ e imagen en $\mathbb{R}$, y generalizar
 los resultados obtenidos para el caso en que haya más de un nivel de
 anidación.


\chapter*{Anexos}
 
 \addcontentsline{toc}{chapter}{Anexos}

 En este anexo se incluyen las demostraciones de las proposiciones
 \ref{prop:sa-anillo-conmutativo-unitario} y
 \ref{prop:SAR-anillo-conmutativo}. 
 
 \setcounter{proposicion}{0}
 \begin{proposicion}
  $\mathbb{SA}$ es un anillo conmutativo y unitario. El neutro con respecto a la 
  suma es el elemento $\mathbf{0} = ((0,0),(0,0))$ y el neutro respecto al producto es 
  $\mathbf{1} = ((1,0),(0,0))$. 
 \end{proposicion}
 \begin{proof}
  Sea $\mathbf{a},\mathbf{b},\mathbf{c}\in\mathbb{SA}$ con
  \begin{displaymath}
	 \mathbf{a} = (\hat{a}_1,\hat{a}_2), \ \ \hat{a}_1,\hat{a}_2\in\mathbb{A},
  \end{displaymath} 
  \begin{displaymath}
	 \mathbf{b} = (\hat{b}_1,\hat{b}_2), \ \ \hat{b}_1,\hat{b}_2\in\mathbb{A},
  \end{displaymath}
  \begin{displaymath}
	 \mathbf{c} = (\hat{c}_1,\hat{c}_2), \ \ \hat{c}_1,\hat{c}_2\in\mathbb{A}.
  \end{displaymath}
  
  Seg'un la Definici'on \ref{def:anillo} lo primero a demostrar es que $\mathbb{SA}$ para la
  suma es un grupo abeliano
  
  Sumando las expresiones
  \begin{displaymath}
	 (\mathbf{a}+\mathbf{b}) + \mathbf{c} = (\hat{a}_1 + \hat{b}_1, \hat{a}_2 + \hat{b}_2) + \mathbf{c}
	                                      = (\hat{a}_1 + \hat{b}_1 + \hat{c}_1 , \hat{a}_2 + \hat{b}_2 + \hat{c}_2),
  \end{displaymath}
  \begin{displaymath}
	 \mathbf{a}+(\mathbf{b} + \mathbf{c}) = \mathbf{a}+(\hat{b}_1 + \hat{c}_1, \hat{b}_2 + \hat{c}_2)
	                                      = (\hat{a}_1 + \hat{b}_1 + \hat{c}_1 , \hat{a}_2 + \hat{b}_2 + \hat{c}_2),
  \end{displaymath}
  se obtiene que
  \begin{displaymath}
	 (\mathbf{a}+\mathbf{b}) + \mathbf{c}= \mathbf{a} + (\mathbf{b}+\mathbf{c}).
  \end{displaymath}
  
  La suma es conmutativa:
  \begin{displaymath}
	 \mathbf{a}+\mathbf{b} = (\hat{a}_1 + \hat{b}_1, \hat{a}_2 + \hat{b}_2) = (\hat{b}_1 + \hat{a}_1, \hat{b}_2 + \hat{a}_2)
	                       = \mathbf{b}+\mathbf{a}.
  \end{displaymath}
  
  El $\mathbf{0}$ es elemento neutro para la suma:
  \begin{displaymath}
	 \mathbf{a}+\mathbf{0} = (\hat{a}_1,\hat{a}_2) + ((0,0),(0,0)) = (\hat{a}_1,\hat{a}_2) = \mathbf{a}.
  \end{displaymath}
  
  Todo elemento tiene opuesto para la suma:  
  Sea $\mathbf{a}'=(- \hat{a}_1,- \hat{a}_2)$, entonces
  \begin{displaymath}
	 \mathbf{a}+\mathbf{a}'= (\hat{a}_1,\hat{a}_2)+(- \hat{a}_1,- \hat{a}_2) = \mathbf{0}.
  \end{displaymath}
  
  Multiplicando las expresiones
  \begin{displaymath}
	 (\mathbf{a}\mathbf{b})\mathbf{c} = (\hat{a}_1 \hat{b}_1,\hat{a}_1\hat{b}_2+\hat{a}_2\hat{b}_1)(\hat{c}_1,\hat{c}_2)
	                                  = (\hat{a}_1 \hat{b}_1 \hat{c}_1,\hat{a}_1 \hat{b}_1\hat{c}_2 +
	                                     \hat{a}_1\hat{b}_2\hat{c}_1+\hat{a}_2\hat{b}_1\hat{c}_1),
  \end{displaymath}
  \begin{displaymath}
	 \mathbf{a}(\mathbf{b}\mathbf{c}) = (\hat{a}_1,\hat{a}_2)(\hat{b}_1 \hat{c}_1,\hat{b}_1\hat{c}_2+\hat{b}_2\hat{c}_1)
	                                  = (\hat{a}_1\hat{b}_1 \hat{c}_1,\hat{a}_1\hat{b}_1\hat{c}_2+
	                                     \hat{a}_1\hat{b}_2\hat{c}_1 + \hat{a}_2\hat{b}_1 \hat{c}_1),
  \end{displaymath}
  se obtiene que
  \begin{displaymath}
	 (\mathbf{a}\mathbf{b})\mathbf{c}=\mathbf{a}(\mathbf{b}\mathbf{c}).
  \end{displaymath}
  
  El producto es conmutativo:  
  \begin{displaymath}
	 \mathbf{a}\mathbf{b}=(\hat{a}_1,\hat{a}_2)(\hat{b}_1,\hat{b}_2)=(\hat{a}_1\hat{b}_1,\hat{a}_1\hat{b}_2+\hat{a}_2\hat{b}_1)
	                     =(\hat{b}_1\hat{a}_1,\hat{b}_1\hat{a}_2+\hat{b}_2\hat{a}_1) = \mathbf{b}\mathbf{a}.
  \end{displaymath}
  
  Multiplicando las siguiente expresiones
  \begin{displaymath}
	 \mathbf{a}(\mathbf{b}+\mathbf{c})=(\hat{a}_1,\hat{a}_2)(\hat{b}_1 + \hat{c}_1, \hat{b}_2 + \hat{c}_2)
	                                  =(\hat{a}_1\hat{b}_1+\hat{a}_1\hat{c}_1,\hat{a}_1\hat{b}_2+\hat{a}_1\hat{b}_2 +
	                                    \hat{a}_2\hat{b}_1 + \hat{a}_2\hat{c}_1),  
  \end{displaymath}
  \begin{displaymath}
	 \mathbf{a}\mathbf{b}+\mathbf{a}\mathbf{c} = (\hat{a}_1 \hat{b}_1,\hat{a}_1\hat{b}_2+\hat{a}_2\hat{b}_1) +
	                                             (\hat{a}_1 \hat{c}_1,\hat{a}_1\hat{c}_2+\hat{a}_2\hat{c}_1)
	                                           = (\hat{a}_1\hat{b}_1+\hat{a}_1\hat{c}_1,\hat{a}_1\hat{b}_2+\hat{a}_1\hat{b}_2 +
	                                              \hat{a}_2\hat{b}_1 + \hat{a}_2\hat{c}_1),   
  \end{displaymath}
  se obtiene que
  \begin{displaymath}
	 \mathbf{a}(\mathbf{b}+\mathbf{c})=\mathbf{a}\mathbf{b}+\mathbf{a}\mathbf{c}.
  \end{displaymath}
  
  El $\mathbf{1}$ es elemento neutro con respecto al producto:   
  \begin{displaymath}
	 \mathbf{a}\mathbf{1} = (\hat{a}_1,\hat{a}_2)((1,0),(0,0)) = (\hat{a}_1,\hat{a}_2) = \mathbf{a}.
  \end{displaymath}
 \end{proof}

 \begin{proposicion}
  El espacio $\mathbb{SA}_{\mathbb{R}}$ es un anillo conmutativo.
 \end{proposicion}
 \begin{proof}
  Sea $\mathbf{a},\mathbf{b},\mathbf{c}\in\mathbb{SA}_\mathbb{R}$ con
  \begin{displaymath}
	 \mathbf{a} = (a_1,\dot{a}_1,a_2,\dot{a}_2), \ \ a_1,\dot{a}_1,a_2,\dot{a}_2\in\mathbb{R},
  \end{displaymath} 
  \begin{displaymath}
	 \mathbf{b} = (b_1,\dot{b}_1,b_2,\dot{b}_2), \ \ b_1,\dot{b}_1,b_2,\dot{b}_2\in\mathbb{R},
  \end{displaymath}
  \begin{displaymath}
	 \mathbf{c} = (c_1,\dot{c}_1,c_2,\dot{c}_2), \ \ c_1,\dot{c}_1,a_2,\dot{c}_2\in\mathbb{R},
  \end{displaymath}
 
  Este espacio con la suma es un grupo abeliano porque coincide con la suma usual de $\mathbb{R}^4$
  como grupo.
  
  Multiplicando las siguiente expresiones
  \begin{displaymath}
	 \mathbf{a}\mathbf{b} = (a_1b_1, a_1\dot{b}_1+\dot{a}_1b_1,a_1b_2+a_2b_1,a_1\dot{b}_2+\dot{a}_1b_2+a_2\dot{b}_1+\dot{a}_2b_1),
  \end{displaymath}
  {\setlength\arraycolsep{2pt}
   \begin{eqnarray*}
    (\mathbf{a}\mathbf{b})\mathbf{c} & = & (a_1b_1c_1,a_1b_1\dot{c}_1 + a_1\dot{b}_1c_1+\dot{a}_1b_1c_1,a_1b_1c_2+a_1b_2c_1+a_2b_1c_1, \\
                                     &   &  a_1b_1\dot{c}_2+a_1\dot{b}_1c_2+\dot{a}_1b_1c_2+a_1\dot{b}_2c_1 +
	                                          \dot{a}_1b_2c_1+a_2\dot{b}_1c_1+\dot{a}_2b_1c_1),  
   \end{eqnarray*}}
   \begin{displaymath}
	  \mathbf{b}\mathbf{c} = (b_1c_1, b_1\dot{c}_1+\dot{b}_1c_1,b_1c_2+b_2c_1,b_1\dot{c}_2+\dot{b}_1c_2+b_2\dot{c}_1+\dot{b}_2c_1),
   \end{displaymath}
   {\setlength\arraycolsep{2pt}
    \begin{eqnarray*}
    \mathbf{a}(\mathbf{b}\mathbf{c}) & = & (a_1b_1c_1,a_1b_1\dot{c}_1 + a_1\dot{b}_1c_1+\dot{a}_1b_1c_1,a_1b_1c_2+a_1b_2c_1+a_2b_1c_1, \\
                                     &   &  a_1b_1\dot{c}_2+a_1\dot{b}_1c_2+\dot{a}_1b_1c_2+a_1\dot{b}_2c_1 +
	                                          \dot{a}_1b_2c_1+a_2\dot{b}_1c_1+\dot{a}_2b_1c_1),  
   \end{eqnarray*}
   }se obtiene que 
  \begin{displaymath}
	 (\mathbf{a}\mathbf{b})\mathbf{c}=\mathbf{a}(\mathbf{b}\mathbf{c}).
  \end{displaymath}
  
  La conmutatividad del producto se obtiene directo de la conmutatividad de la suma y la multiplicaci'on
  de los n'umeros reales.
  
  Multiplicando las expresiones
  \begin{displaymath}
	 \mathbf{b}+\mathbf{c} = (b_1+c_1,\dot{b}_1+\dot{c}_1,b_2+c_2,\dot{b}_2+\dot{c}_2),
  \end{displaymath}
  {\setlength\arraycolsep{2pt}
   \begin{eqnarray*}
	  \mathbf{a}(\mathbf{b}+\mathbf{c}) & = & (a_1b_1+a_1c_1,a_1\dot{b}_1+a_1\dot{c}_1+\dot{a}_1b_1+\dot{a}c_1,
	                                           a_1b_2+a_1c_2+a_2b_1+a_2c_1, \\
	                                    &   &	 a_1\dot{b}_2+a_1\dot{c}_2+ \dot{a}_1b_2+\dot{a}_1c_2,
	                                           a_2\dot{b}_1+a_2\dot{c}_1+\dot{a}_2b_1+\dot{a}_2c_1),
   \end{eqnarray*}}
  \begin{displaymath}
	 \mathbf{a}\mathbf{b} = (a_1b_1, a_1\dot{b}_1+\dot{a}_1b_1,a_1b_2+a_2b_1,a_1\dot{b}_2+\dot{a}_1b_2+a_2\dot{b}_1+\dot{a}_2b_1),
  \end{displaymath}
  \begin{displaymath}
	 \mathbf{a}\mathbf{c} = (a_1c_1, a_1\dot{c}_1+\dot{a}_1c_1,a_1c_2+a_2c_1,a_1\dot{c}_2+\dot{a}_1c_2+a_2\dot{c}_1+\dot{a}_2c_1),
  \end{displaymath}
  {\setlength\arraycolsep{2pt}
   \begin{eqnarray*}
	  \mathbf{a}\mathbf{b}+\mathbf{a}\mathbf{c} & = & (a_1b_1+a_1c_1,a_1\dot{b}_1+a_1\dot{c}_1+\dot{a}_1b_1+\dot{a}c_1,
	                                                   a_1b_2+a_1c_2+a_2b_1+a_2c_1, \\
	                                            &   &	 a_1\dot{b}_2+a_1\dot{c}_2+ \dot{a}_1b_2+\dot{a}_1c_2,
	                                                   a_2\dot{b}_1+a_2\dot{c}_1+\dot{a}_2b_1+\dot{a}_2c_1),
   \end{eqnarray*}
  }se obtiene que
  \begin{displaymath}
	 \mathbf{a}(\mathbf{b}+\mathbf{c})=\mathbf{a}\mathbf{b}+\mathbf{a}\mathbf{c}
  \end{displaymath}
 \end{proof}

\end{document}